\documentclass{article}
\usepackage[a4paper]{geometry}
\usepackage{setspace}
\usepackage{authblk}
\usepackage{setspace}
\usepackage[utf8]{inputenc} 
\usepackage[T1]{fontenc}    
\usepackage{hyperref}       
\usepackage{url}            
\usepackage{booktabs}       
\usepackage{amsmath}
\usepackage{amsfonts}
\usepackage{amsthm}
\usepackage{amssymb}
\usepackage{bbm}
\usepackage{tikz}
\usepackage{nicefrac}       
\usepackage{microtype}      
\usepackage{lipsum}
\usepackage{fancyhdr}       
\usepackage{graphicx}       
\usepackage{algorithm}
\usepackage{algpseudocode}
\usepackage{mathtools}
\usepackage{natbib}
\usepackage{amsmath}
\usepackage{authblk}
\usepackage{cleveref}
\usepackage{verbatim}
\graphicspath{{figures/}}     
\usepackage{mathtools}

\newtheorem{theorem}{Theorem}[section]
\newtheorem{lemma}{Lemma}[section]
\newtheorem{corollary}{Corollary}[section]
\newtheorem{proposition}{Proposition}[section]

\newtheorem{assumption}{Assumption}

\theoremstyle{definition}

\newtheorem{remark}{Remark}

\DeclareMathOperator*{\argmin}{argmin}

\title{Latent space models for grouped multiplex networks}

\author[1]{Alexander Kagan \thanks{Corresponding author. Email: \texttt{amkagan@umich.edu}}}
\author[2]{Peter W. MacDonald}
\author[1]{Elizaveta Levina}
\author[1]{Ji Zhu}

\affil[1]{Department of Statistics, University of Michigan}
\affil[2]{Department of Statistics \& Actuarial Science, University of Waterloo}

\date{}

\begin{document}
\maketitle

\begin{abstract}
Complex multilayer network datasets have become ubiquitous in various  applications, such as neuroscience, social sciences, economics, and genetics. Notable examples include brain connectivity networks collected across multiple patients or trade networks between countries collected across multiple goods. Existing statistical approaches to such data typically focus on modeling the structure shared by all networks; some go further by accounting for individual, layer-specific variation. However, real-world multilayer networks often exhibit additional patterns shared only within certain subsets of layers, which can represent treatment and control groups, or patients grouped by a specific trait. Identifying these group-level structures can uncover systematic differences between groups of networks and influence many downstream tasks, such as testing and low-dimensional visualization.
To address this gap in existing research, we introduce the GroupMultiNeSS model, which generalizes the previously proposed MultiNeSS model of \cite{multiness}. This model enables the simultaneous extraction of shared, group-specific, and individual latent structures from a sample of multiplex networks — multiple, heterogeneous networks observed on a shared node set. For this model, we establish identifiability, develop a fitting procedure using convex optimization in combination with a nuclear norm penalty, and prove a guarantee of recovery for the latent positions as long as there is sufficient separation between the shared, group-specific, and individual latent subspaces. We compare the model with MultiNeSS and other models for multiplex networks in various synthetic scenarios and observe an apparent improvement in the modeling accuracy when the group component is accounted for. Experiment with the Parkinson's disease brain connectivity dataset of \cite{TaoWu2017} demonstrates the superiority of GroupMultiNeSS in highlighting node-level insights on biological differences between the treatment and control patient groups.
\end{abstract}

\section{Introduction}
    \label{sec:introduction}
    Complex multilayer network datasets have become ubiquitous in many area of applications, including neuroscience, genetics, social sciences, and economics, among others.  Examples of such datasets include brain connectivity networks observed in a group of patients, protein-protein interaction networks observed across multiple tissues, or trade networks in multiple commodities between countries; in all these cases, the same set of nodes is shared across multiple networks (layers).   A number of statistical methods for such data have been proposed, including natural generalizations of single-layer models, such as the stochastic block model (SBM) \citep{Holland1983}, the random dot product graph (RDPG) \citep{Young2007, Athreya2017}, and the  latent variable model of \cite{MaMaUnifModel2020} to their respective multilayer versions \citep{han2015consistent, Jones2021MRDPG, xuefei2020flexible}. Other lines of work focused on Bayesian approaches to latent space models \citep{gollini2016joint, SalterTownshend2017, DAnegelo2019, Sosa2021}, and models based on low-rank assumptions, such as the Common Subspace Independent Edge Model (COSIE) \citep{Arroyo2019} and the MultiNeSS model \citep{multiness, TianYinqiu2024}. 
All of these models, in various ways, focus on how to model the common structure shared by all the layers, while accounting for individual layer-specific information. For example, the multilayer SBM assumes the communities to be the same across layers but allows their connection probabilities to differ, COSIE assumes all expected layer adjacency matrices lie in a common subspace, and MultiNeSS decomposes each layer's latent space into shared and individual subspaces. 

While the common structure is the natural starting point for studying multilayer networks, in many applications there are important questions that go beyond common and individual structure.   
For instance, in the context of brain connectivity, one may want to look at the differences between patients and healthy controls, or between a treatment and a placebo, while separating out both the common structure shared by all humans and the individual variability irrelevant to the scientific question at hand.  In this work, we study the question of how to estimate structure specific to groups of network layers within the collection, and how to separate these group-level structures from the structure shared by all the layers. 

Typical studies focused on discovering abnormal connectivity patterns in patients with a given disease, such as ADHD \citep{DiffADHD2017}, Alzheimer's \citep{DiffAlzheimer2024}, or Parkinson's \citep{TaoWu2017}, tend to perform permutation tests on descriptive summary statistics of the network groups, such as average connectivity or average path length.   These  descriptive statistics are often insufficient to account for the full complexity of the networks, resulting in low power of the tests, and are not able to separate out the common structure.  Latent space network models have been also used to develop classical two-sample hypothesis tests, testing whether the latent positions of nodes in the two samples are drawn from the same distribution, usually up to an orthogonal rotation or scaling \citep{TangSemiparamNetworkTest, nguen2024networktwosampletestblock, macdonald2024mesoscale}.   While these approaches provide valid statistical tests, they are also unable to separate out the common structure before comparing groups, resulting in lower power.  Further, both these types of methods do not provide estimates of group structure that could be visualized and compared.  


In contrast, here we focus on developing a latent space model specifically for multiplex networks whose layers can be divided into groups.   It generalizes the previously proposed MultiNeSS model by allowing the latent positions of nodes within each layer to comprise not only a component shared across all networks and an individual component specific only to this layer, but also an additional group component shared by the networks only within its group. 
We develop a computationally efficient fitting algorithm for this model, and establish theoretical guarantees under the Gaussian edges model. Through simulation studies and applications to real data, we show that incorporating a group component significantly improves modeling accuracy. Our results highlight the importance of accounting for group-level structure in latent space models and provide a versatile tool for analyzing grouped network data.

The rest of this manuscript is organized as follows. Section \ref{sec:model} presents the model and states the identifiability conditions for its parameters. In Section \ref{sec:estimation}, we describe the fitting procedure and propose a cross-validation approach for selecting hyperparameters.  In Section \ref{sec:theory}, we
establish consistency for the fitting algorithm under the Gaussian assumption on edge values. In Section \ref{sec:experiments},  we present numerical simulations on synthetic data to illustrate the efficacy of our method, and in Section \ref{sec:real_data}, we apply it to the Parkinson's disease brain connectivity data from \citet{TaoWu2017}.  Section \ref{sec:discussion} discusses conclusions and future work.  



 
\section{A model for grouped multiplex networks}
    \label{sec:model}
	
We start by fixing notation. The observed data are $M$ undirected networks on a shared set of $n$ nodes, separated into $K$ groups of sizes $m_k$, $k =1, \ldots, K,$ so that $M = \sum_{k=1}^K m_k$. 
Each network is represented by a symmetric and possibly weighted adjacency matrix 
$$A_{k \ell} \in \mathbb{R}^{n\times n}, \ (k, \ell) \in \mathcal{I}, \quad \text{where} \quad \mathcal{I} :=\{(k, \ell): k=1,\ldots, K; \ \ell=1,\ldots, m_k\},$$ 
where the edge value $A_{k \ell, ij}$ represents the strength of connection between nodes $i$ and $j$ in network $\ell$ of group $k$. Each node $i$, $i = 1, \dots, n$ 
is associated with a  layer-specific fixed latent position $x_{k \ell}^{(i)}\in \mathcal{X}_{k\ell} \subset \mathbb{R}^{D_{k\ell}}$, which we stack into a matrix $X_{k \ell} \in \mathbb{R}^{n\times D_{k\ell}}$.  We assume that conditional on $X_{k \ell}$, the elements of $A_{k \ell}$ for $i \le j$ are independently drawn and follow an edge entry distribution 
\begin{equation}\label{general_edge_distrib}
    A_{k \ell, ij} \stackrel{\text{ind}}{\sim} f(\cdot;\kappa_{k\ell}(x_{k \ell}^{(i)}, x_{k \ell}^{(j)}), \phi_{k\ell}), \quad 1\le i \le j\le n, \quad (k, \ell) \in \mathcal{I}, 
\end{equation}
where $\kappa_{k\ell}:\mathcal{X}_{k\ell}\times \mathcal{X}_{k\ell} \rightarrow \mathbb{R}$ is a symmetric function capturing similarity between the two input latent positions, and $\phi_{k\ell}$ is a possible nuisance parameter of the edge distribution $f$. If the modeled networks are free of self-loops, we restrict this assumption to the entries with $i < j$. Many previously proposed multiplex network models have this form;  for example, if $\kappa_{k\ell}$ is the inner product, $\mathcal{X}_{k\ell}$ is such that $0 \le x^\top y \le 1$ for all $x, y\in \mathcal{X}_{k\ell}$, and $f$ corresponds to the Bernoulli distribution, we get the Multilayer RDPG model \citep{Jones2021MRDPG}. 

We assume that the edge distribution belongs to a canonical exponential family of the form 
\begin{equation}\label{one_param_exp_family}
    f(x;\theta)\propto \exp\{\theta x - \nu(\theta)\}, \quad \theta \in \mathbb{R}
\end{equation}
with a natural parameter $\theta$ 
and a log-partition function $\nu$.  This family includes Bernoulli, Poisson, Gaussian, and  many other distributions, and is suitable for modeling many different edges types across a range of applications, including binary, count, and continuous edge values.

The distribution of edges in layer $A_{k\ell}$ in  \eqref{general_edge_distrib} depends on the latent positions $X_{k\ell}$ only through the Gram matrix $\Theta_{k\ell} \in \mathbb{R}^{n\times n}$ with $\Theta_{k\ell, ij}=\kappa_{k\ell}(x_{kl}^{(i)}, x_{k\ell}^{(j)})$; this will naturally raise questions about identifiability, addressed in Section \ref{sec:identifiability}.  In particular, when $\kappa_{k\ell}$ is the Euclidean inner product, $\Theta_{k\ell} = X_{k\ell}X_{k\ell}^\top$, the key low rank matrix  for many network latent space models starting from the seminal work of \citep{Hoff2002}.    To emphasize this natural parameter, we can rewrite \eqref{general_edge_distrib} in matrix as follows:
\begin{equation}\label{general_edge_disrib_matrix_form}
    A_{k\ell}\ \stackrel{\text{ind}}{\sim} f(\cdot; \Theta_{k\ell}, \phi_{k\ell}), \quad (k,\ell) \in \mathcal{I}.
\end{equation}

\subsection{A group latent space structure}
Here we propose a new latent space model of the general type   \eqref{general_edge_disrib_matrix_form} which explicitly includes common, group-level, and individual structure.   We call it GroupMultiNeSS  (GROUPed MULTIplex NEtworks with Shared Structure).   The key assumption of GroupMultiNeSS, illustrated in Figure \ref{fig:latent_space_gmn}, is that the latent positions $X_{k \ell} \in 
\mathbb{R}^{n\times D_{k\ell}}$ can be decomposed into parts representing 
 common structure $V\in\mathbb{R}^{n \times d_{0}}$ shared by all the layers, group-specific structure $W_k \in\mathbb{R}^{n \times d_{k}}$ shared by the layers within a given group $k=1, \ldots,K$, and finally, individual structure $U_{k \ell} \in\mathbb{R}^{n \times d_{k \ell}}$ specific to a given layer $\ell=1, \ldots, m_k$ within group $k$;  that is, 
\begin{equation}\label{gmn_decomposition}
    X_{k \ell} = [V \ W_{k} \ U_{k \ell}], \quad (k, \ell)\in\mathcal{I},
\end{equation}
and $D_{k\ell} = d_0 + d_k +d_{k\ell}$.
If the number of groups is set to 1, this model recovers the MultiNeSS model of \cite{multiness} as a special case.   

\begin{figure}[t!]
    \centering
    \includegraphics[width=0.4\linewidth]{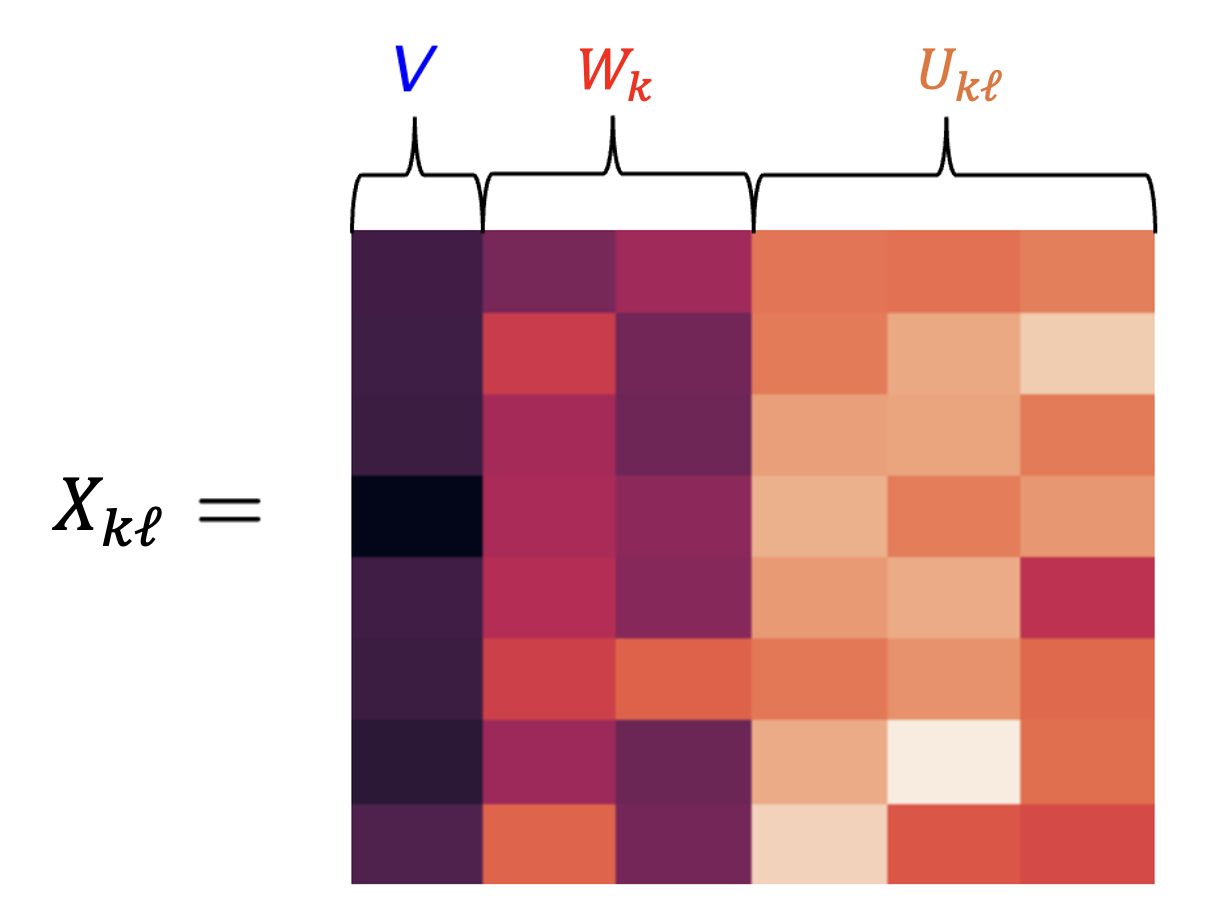} \qquad 
    \includegraphics[width=0.4\linewidth]{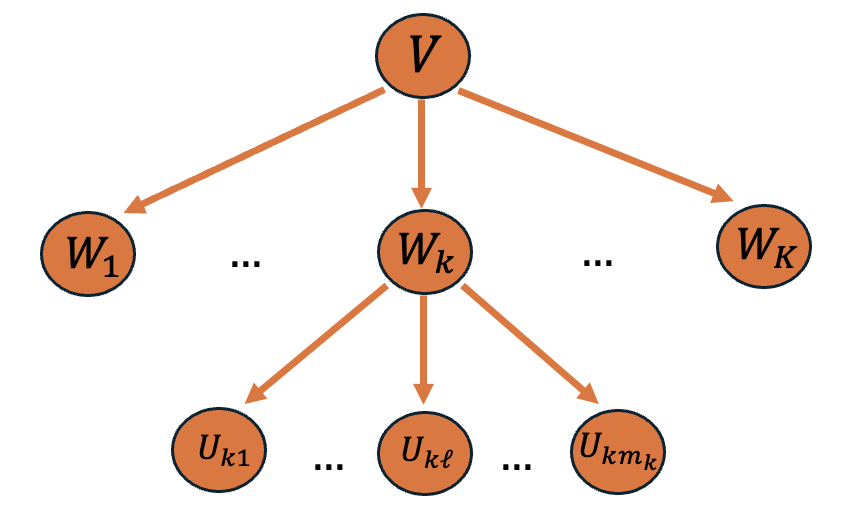}
    \caption{Latent space decomposition assumed by GroupMultiNeSS.}
    \label{fig:latent_space_gmn}
\end{figure}

We set the similarity function $\kappa_{k\ell}$ to be the generalized inner product, defined as 
\begin{equation}\label{generalized_inner_product}
\kappa_{p, q}(x, y)=x_1 y_1+\cdots+x_p y_p-x_{p+1} y_{p+1}-\cdots-x_{p+q} y_{p+q}=x^{\top} I_{p, q} y,
\end{equation}
where $I_{p, q}$ is a $(p+q)\times (p+q)$ diagonal matrix with first $p$ diagonal values equal to $1$ and remaining $q$ equal to $-1$.   The first $p$ dimensions are referred to as assortative, meaning that higher similarity in this dimension increases the edge weight, and the last $q$ dimensions are disassortative, meaning a higher similarity decreases the edge weight.   We allow the number of assostative and disassortative dimensions to be different in each latent component, assuming that $\Theta_{k\ell}$ has the general form 
\begin{equation}\label{gmn_gram_matrix}
    \Theta_{k\ell}  = VI_{p_0, q_0}V^\top + W_k I_{p_k, q_k} W_k ^\top + U_{k \ell}I_{p_{k \ell}, q_{k \ell}} U_{k \ell}^\top,
\end{equation}
where $p_k + q_k = d_k$ for $k = 0, 1, \ldots,K$ and $p_{k \ell} + q_{k \ell} = d_{k \ell}$ for $(k, \ell)\in\mathcal{I}$. 



\subsection{Identifiability}\label{sec:identifiability}
In this section, we give a sufficient condition for the identifiability of the GroupMultiNeSS parameters, for a somewhat more general similarity function $\kappa(x, y) = \psi(x^\top I_{p, q}y)$, an invertible scalar function of the generalized inner product defined in \eqref{generalized_inner_product}. This is an easy generalization since if 
$f(\cdot;\theta_1) \neq f(\cdot;\theta_2)$ for any $\theta_1 \ne \theta_2$, the distributions  will also not coincide whenever $\psi(\theta_1)\ne \psi(\theta_2)$.  

First consider the case of $M= 1$, a single network.  The standard identifiability condition for an inner product latent position model is linearly independent columns of the latent position matrix $X_{k\ell} = [V \ W_k \ U_{k\ell}]$.   However, this is evidently not enough for GroupMultiNeSS, since we need to observe at least two groups to separate $V$ from $W_k$ and at least two networks in each group to be able to separate $W_k$ from $U_{k\ell}$.   We will thus assume $K \ge 2$ and $m_k \ge 2$ for all $k$.  
Further, intuitively we need $V$ to contain all the information shared by all the layers, while $W_k$ should contain all of the information shared by the layers in group $k$ but not overlap with $V$, and similarly for $W_k$ and $U_{k\ell}$. Finally, as in all inner product models, each latent component is only identifiable up to an indefinite orthogonal rotation, defined by $\mathcal{O}_{p, q} = \{O \in \mathbb{R}^{(p + q) \times (p+q)}: O^\top I_{p, q}O = I_{p, q}\}$.  
The following proposition formalizes this intuition.   The proof can be found in section \ref{identif_proof_section} of Appendix. 
\begin{proposition}\label{identif_propos}
Assume $\{f(\cdot;\theta, \phi), \ \theta \in \mathbb{R}\}$ is an identifiable parametric family and $\kappa(x, y) = \psi(x^\top I_{p, q}y)$ is an invertible function of the generalized
inner product.  Assume the following conditions hold:  
\begin{enumerate}
    \item For each $(k, \ell) \in \mathcal{I}$, columns of the matrix $X_{k \ell} = \left[V \ W_k \ U_{k\ell}\right]$ are linearly independent.
    \item For each $k=1, \ldots, K$, there exist $\ell_1 \neq \ell_2$ 
    such that the columns of the matrix 
    $\left[V \ W_k \ U_{k \ell_1} \ U_{k \ell_2}\right]
    $
    are linearly independent.
    \item There exist $(k_1, \ell_1), (k_2, \ell_2) \in\mathcal{I}$ with $k_1 \ne k_2$ such that
    the columns of the matrix \\ 
    $[V \ W_{k_1} \ W_{k_2} \ U_{k_1 \ell_1} \ U_{k_2 \ell_2}]$
    are linearly independent.
\end{enumerate}
Then the parameters of the GroupMultiNeSS model \eqref{general_edge_distrib} are identifiable up to indefinite orthogonal transformations, that is, if the probability distributions induced by two different parameterizations $\left(V, \{W_k\}_{k=1}^K, \{U_{k\ell}\}_{\mathcal{I}}\right)$ and $\left(V', \{W_k'\}_{k=1}^K, \{U_{k\ell}'\}_{\mathcal{I}}\right)$ coincide, then
\begin{align*}
     V=V^{\prime} O_0, \ \ 
    W_k = W_k'O_{k},  \ \   
     U_{k\ell}=U_{k\ell}^{\prime} O_{k\ell}, 
\end{align*}
for some indefinite orthogonal rotations $O_k \in \mathcal{O}_{p_k, q_k}, k=0, \ldots, K$ and $O_{k\ell} \in \mathcal{O}_{p_{k\ell}, q_{k\ell}}, (k, \ell)\in\mathcal{I}$.
\end{proposition}
\begin{remark}
    Intuitively, Condition 2 requires that in each group $k$ there are at least two individual components with columns linearly independent of those shared by the whole group ($[V, \ W_k]$) -- if that is not the case, then the individual components still contain some common group information. Similarly, Condition 3 requires that the shared component $V$ has linearly independent columns from at least some $[W_{k_1},\ U_{k_1\ell_1}]$ and $[W_{k_2},\ U_{k_2\ell_2}]$, as otherwise not all shared information would have been separated.  
\end{remark}


\section{Fitting the GroupMultiNeSS model}
    \label{sec:estimation}
	The true latent positions $V, \{W_k\}_{k=1}^K, \{U_{k\ell}\}_{\mathcal{I}}$ are only identifiable  up to a rotation.  To avoid making an arbitrary choice of the rotation, we instead focus on estimating the Gram matrices 
\begin{align}\label{gmn_product_notation}
    S = VI_{p_0, q_0}V^\top, \ \ 
    Q_{k} = W_kI_{p_k, q_k}W_k^\top, \ \ 
    R_{k \ell} = U_{k \ell}I_{p_{k\ell}, q_{k\ell}}U_{k \ell}^\top, 
\end{align}
for $k=1,\ldots, K$, $(k, \ell)\in\mathcal{I}$.   If needed, we can always extract estimates of the latent positions themselves by performing an SVD on the estimated Gram matrices, and fix a rotation to align latent positions across groups.  

\subsection{Likelihood maximization with a nuclear norm penalty}
A  natural way to estimate parameters \eqref{gmn_product_notation} is to minimize the negative log-likelihood under the GroupMultiNeSS model, which is easy to write down due to independence.  We can write the full data likelihood as 
\begin{align}\label{log_likelihood}
    \mathcal{L}\left(S, \left\{Q_k\right\}_{k = 1}^K, \left\{R_{k\ell}\right\}_{\mathcal{I}}; \ \left\{A_{k \ell}\right\}_{\mathcal{I}}\right)=\sum_{k=1}^K \mathcal{L}_k(S + Q_k, \{R_{k\ell}\}_{\ell=1}^{m_k}; \ \left\{A_{k \ell}\right\}_{\ell=1}^{m_k}) \, , 
\end{align}
where 
\begin{equation}\label{group_log_likelihood}
        \mathcal{L}_k(S + Q_k, \{R_{k\ell}\}_{\ell=1}^{m_k}; \ \left\{A_{k \ell}\right\}_{\ell=1}^{m_k}) = -\sum_{\ell=1}^{m_k} \sum_{i\le j} \log f\left(A_{k \ell, i j} ; S_{i j}+Q_{k, i j} + R_{k \ell, ij}, \phi\right). 
\end{equation}
The sum over $i \le j$ can be replaced with the sum over $i < j$ if we do not want to model diagonal entries.

Imposing low-rank constraints on the matrices defined \eqref{gmn_product_notation}  turns the optimization problem \eqref{log_likelihood} non-convex and typically intractable.  
Instead, we follow the common strategy of replacing the rank constraints with a nuclear norm penalty, which gives us a convex optimization problem as long as the edge distribution density $f(\cdot; \theta, \phi)$ is log-concave in $\theta$.

Note that the group log-likelihood $\mathcal{L}_k$ depends on $S + Q_k$ and $\{R_{k\ell}\}_{\ell=1}^{m_k}$ for that $k$ only, allowing for a two-stage approach to solving the optimization problem:   we can first estimate $\widehat{S + Q_k}$ and $\{\hat{R}_{k\ell}\}_{\ell=1}^{m_k}$ within each group and then estimate $\hat{S}$ and $\hat{Q}_k, k=1,  \ldots, K$ by optimizing the full likelihood $\mathcal{L}$ with all $\{\hat{R}_{k\ell}\}_\mathcal{I}$ fixed.    Formally, in the first stage we solve the optimization problem
\begin{equation}\label{first-stage-conv-prob}
\begin{aligned}
    \min_{S + Q_k, R_{k \ell}}\Bigl\{& \mathcal{L}_k\left(S + Q_k,\left\{R_{k\ell}\right\}_{\ell=1}^{m_k}; \ \left\{A_{k \ell}\right\}_{\ell=1}^{m_k}\right)+ \lambda_{1k}\|S + Q_k\|_*+\sum_{\ell=1}^{m_k}  \lambda_{1k}\alpha_{1k \ell}\left\|R_{k \ell}\right\|_* \Bigr\}
\end{aligned}
\end{equation}
and in the second stage, we solve 
\begin{equation}\label{second-stage-conv-prob}
   \min_{S, Q_k}\Bigl\{\mathcal{L}\left(S, \left\{Q_k\right\}_{k=1}^K,\bigl\{\hat{R}_{k\ell}\bigr\}_{\mathcal{I}}; \ \left\{A_{k \ell}\right\}_{\mathcal{I}}\right)+\lambda_2\|S\|_*+\sum_{k=1}^K\lambda_2 \alpha_{2k}\left\|Q_{k }\right\|_* \Bigr\}.
\end{equation}
Here, $\|\cdot\|_*$ denotes the matrix nuclear norm (the sum of its singular values), and $ \{\lambda_{1k}\}_{k=1}^{K}, \lambda_2$, $\{\alpha_{2k}\}_{k=1}^K,\{\alpha_{1k \ell}\}_{\mathcal{I}}$ are hyperparameters;  choosing them in practice is discussed in Section \ref{hyperparam_tuning_section}).

\begin{remark}
    The two-stage fitting procedure can be thought of as a bottom-to-top fitting of the GroupMultiNeSS tree in Figure \ref{fig:latent_space_gmn}:  the first stage estimates the "leaves" at the bottom, which stay fixed after, and the second stage estimates the next level of the tree. This procedure is thus easily extended to nested groups, as long as they are arranged in a hierarchical tree.   
\end{remark}
We solve the optimization problems \eqref{first-stage-conv-prob} and \eqref{second-stage-conv-prob} by block coordinate descent, updating one of the matrices at a time while keeping the others fixed.   Each matrix update is performed via a proximal gradient method developed for the nuclear norm penalty by \citep{FithianMazumder2018}. Given the gradient step size $\eta > 0$, the update at iteration $t \ge 1$ is given by 
\begin{equation}\label{proximal_map_optim}
    Z^{(t+1)} = \argmin_{Z\in\mathbb{R}^{n\times n}}\Bigl\{ {1\over 2\eta} \Bigl\|Z-\Bigl[Z^{(t)} - \eta \nabla L(Z^{(t)})\Bigr] \Bigr\|_F^2 + \rho \|Z\|_* \Bigr\} \ , 
\end{equation}
where $L$, $Z$, and $\rho$ stand for generic loss function, matrix being optimized over, and the nuclear norm penalty tuning parameter,  respectively.  It is well-known that the solution to the optimization problem in \eqref{proximal_map_optim} is given by the soft-thresholding operator, 
\begin{equation}\label{proximal_map_operator}
    Z^{(t+1)} 
 = \mathcal{T}_{\eta\rho}\left(Z^{(t)} - \eta \nabla L(Z^{(t)})\right) , 
\end{equation}
where for a  matrix with the singular value decomposition $Z=U\Sigma V^\top$ the soft-thresholding operator $\mathcal{T}_s$  is defined as
$$\mathcal{T}_s(Z) =  U \operatorname{diag}[(\Sigma_{11} - s)_+, \ldots, (\Sigma_{nn} - s)_+]V^\top \, , 
$$
with $(x)_+ := \max\{0, x\}$.   
Intuitively, the soft-thresholding operator projects the previous iterate onto the space of low-rank matrices by truncating its singular values.  

If the true rank $d$ of $Z$ is known (oracle estimator), soft thresholding is usually replaced by hard thresholding, 
\begin{equation}\label{proximal_map_operator_hard}
Z^{(t+1)} 
 = \left[Z^{(t)} - \eta \nabla L(Z^{(t)})\right]_d
\end{equation}
where  
$ [Z]_d:=\argmin_{\operatorname{rank}\left(Z^{\prime}\right) \leq d}\left\|Z-Z^{\prime}\right\|_F$, 
 by the Eckart-Young theorem, is the truncation of the SVD of $Z$ to the largest $d$ singular values.


We summarize the entire optimization procedure in Algorithm \ref{alg:stepwise_update}. In the first stage, we solve problem \eqref{first-stage-conv-prob} for every $k=1, \ldots, K$ by alternating between updates of $S+Q_k$ and $\{R_{k\ell}\}_{l=1}^{m_k}$ until convergence. In the second stage, we solve problem \eqref{second-stage-conv-prob} by alternating between updates of $S$ and $\{Q_k\}_{k=1}^K$. Each matrix update is performed according to \eqref{proximal_map_operator} but with stage-specific learning rates  $\eta_1,\eta_2 > 0$ and parameter-specific normalization for stability:   we normalize $\eta_1$ by $m_k$ for updates of $S + Q_k$ and  $\eta_2$ by $M$ and $m_k$ for updates of $S$ and $Q_k$, respectively, that is, we divide by the number of times these parameters appear in the joint log-likelihood \eqref{log_likelihood}. The parameter matrices obtained by running the alternating updates within each optimization subproblem are passed to the optional {\it refitting step}, a common post-processing step for the spectral regularization methods \citep{mazumder}. This procedure ``unshrinks'' non-zero eigenvalues of the fitted matrices to compensate for excessive truncation caused by the nuclear norm penalty (see details in Section \ref{refitting_details_section} of the Appendix). The first and second stage refitting procedures differ because the individual components $\hat{R}_{k\ell}$ from the first stage are kept fixed in the second stage; we refer to the algorithms as FirstStageRefit and SecondStageRefit, respectively.

We postpone the discussion of initialization to Section \ref{init_section} and for now state the algorithm for generic initialization algorithms FirstStageInit and SecondStageInit.  

\begin{remark}
The main computational bottleneck within every iteration of the first and second stages of Algorithm \ref{alg:stepwise_update} is computing the SVD for each updated parameter matrix, which is needed for soft thresholding. Computing the truncated SVD for $r$ leading singular values gives the total iteration complexity of $O(m_kr n^2)$ for Stage I and $O(Kr n^2)$ for Stage II. In our implementation, we take $r=\sqrt{n}$ as the default value. We also parallelize the first stage computations across groups.    
\end{remark}

\begin{remark}  One alternative approach to optimizing the joint log-likelihood in \eqref{log_likelihood} would be to add nuclear norm penalties for all of the matrices $S, Q_k, R_{k\ell}$ and update all of them in a single loop.  An important advantage of Algorithm \ref{alg:stepwise_update} over this is that the loop over groups $k = 1,\ldots, K$ can be  parallelized.  While we then have to perform a second stage optimization, we expect it to be much less computationally expensive than each first-stage problem, because the number of parameter matrices in the second stage is $K + 1$, which is typically much smaller than $m_k + 1, k=1, \ldots, K$ for each of the parallelized optimizations in the first stage.   Empirically, we indeed observed that this approach is not only significantly more expensive in terms of CPU time and the number of SVD computations, but also less stable in terms of convergence. 
\end{remark}


\begin{remark}
Another alternative is to apply a non-convex approach analogous to that of \cite{TianYinqiu2024}, which pre-estimates the latent component ranks via the Shared Space Hunting (SSH) approach; we could apply SSH and then replace soft with hard thresholding.  
The key advantage of this approach is avoiding cross-validation, and estimating ranks independently of the assumed likelihood.  However, SSH is only guaranteed to accurately estimate ranks that are much smaller than $n$, and we have empirically observed it is often unstable in harder settings.  When the ranks are misspecified, our Algorithm \ref{alg:stepwise_update} performs significantly better, and it gives comparable or marginally worse results when ranks are estimated correctly.   While both are valid approaches, we have opted to use cross-validation for stability, in spite of its higher computational cost.  

\end{remark}

The optimization procedure in the first stage \eqref{first-stage-conv-prob} essentially applies the MultiNeSS fitting algorithm to layers $\{A_{k\ell}\}_{\ell=1}^{m_k}$ (as defined in Equation 5 of \citet{multiness}),  aiming to separate the individual components $\{R_{k\ell}\}_{\ell=1}^{m_k}$ from the  $S+Q_k$, the shared component for group $k$. In contrast, the second-stage optimization is not directly interpretable as an application of MultiNeSS, since we fix the estimated individual components. To provide  some intuition about the second stage of Algorithm \ref{alg:stepwise_update}, we explicitly state its relationship to the MultiNeSS algorithm in the special case of the Gaussian edge-entry distribution, showing that it essentially fits the MultiNeSS model to the group-wise averaged residuals between the layers and their estimated individual components. The proof of this proposition can be found in Section \ref{fit_section_proofs} of the Appendix.

\begin{algorithm}[t!]
\caption{Stepwise proximal gradient descent for optimizing \eqref{log_likelihood} \label{alg:stepwise_update}}
\begin{algorithmic}
    \State \textbf{Input:} Adjacency matrices $\{A_{k\ell}\}_\mathcal{I}$, penalty coefficients $\{\lambda_{1k}\}_{k=1}^K, \lambda_2, \{\alpha_{2k}\}_{k=1}^K, \{\alpha_{1k\ell}\}_\mathcal{I}$, learning rates $\eta_1, \eta_2 > 0$.
    \State \textbf{Output:} $\hat{S}$, $\{\hat{Q}_k\}_{k=1}^K$, $\{\hat{R}_{k\ell}\}_\mathcal{I}$
    \\
    \State \underline{\bf Stage I}
    \For {$k=1, \ldots, K$}
        \State $(S + Q_k)^{(0)}, \{R^{(0)}_{k\ell}\}_{\ell=1}^{m_k} \gets \operatorname{FirstStageInit}(\{A_{k\ell}\}_{\ell=1}^{m_k})$
       \For{iteration $t = 1, 2,...$  until convergence}
            $$
            \begin{array}{ll}
                &R_{k\ell}^{(t)} \gets \mathcal{T}_{\eta_1\lambda_{1k}\alpha_{1k\ell}} \left[ R_{k\ell}^{(t-1)} - \eta_1\frac{\partial}{\partial R_{k\ell}} \mathcal{L}_k( (S + Q_k)^{(t-1)}, \{R_{k\ell}^{(t-1)}\}) \right],\quad \text{ for } \ell = 1, \ldots, m_k,\\
                &(S + Q_k)^{(t)} \gets \mathcal{T}_{\eta_1\lambda_{1k}/m_k} \left[ (S + Q_k)^{(t-1)} - \frac{\eta_1}{m_k} \frac{\partial}{\partial Q_k} \mathcal{L}_k((S + Q_k)^{(t-1)}, \{R_{k\ell}^{(t)}\}) \right].
            \end{array}
            $$
        \EndFor
        \State $\widehat{S + Q_k}, \{\hat{R}_{k\ell}\}_{\ell=1}^{m_k} \gets \operatorname{FirstStageRefit}\bigl((S+Q_k)^{(t)}, \{R_{k\ell}^{(t)}\}_{\ell=1}^{m_k}; \{A_{k\ell}\}_{\ell=1}^{m_k}\bigr)$
    \EndFor\\
    \State \underline{\bf Stage II}
    \State $S^{(0)}, \{Q_k^{(0)}\}_{k=1}^K \gets \operatorname{SecondStageInit}(\{A_{k\ell}\}_{\mathcal{I}}, \{\hat{R}_{k\ell}\}_\mathcal{I}, \{\widehat{S + Q_k}\}_{k=1}^{K})$
    \For{iteration $t = 1, 2,... $ until convergence}
            \State $$
            \begin{array}{lll}
                &Q_k^{(t)} \gets \mathcal{T}_{\eta_2\lambda_2\alpha_{2k}/m_k} \left[ Q_k^{(t-1)} - \frac{\eta_2}{m_k} \frac{\partial}{\partial Q_k} \mathcal{L}(S^{(t-1)}, \{Q_k^{(t-1)}\}, \{\hat{R}_{k\ell}\}) \right], & \text{ for } k=1, \ldots, K,\\
                &S^{(t)} \gets \mathcal{T}_{\eta_2\lambda_2/M} \left[ S^{(t-1)} - \frac{\eta_2}{M} \frac{\partial}{\partial S} \mathcal{L}(S^{(t-1)}, \{Q_k^{(t)}\}, \{\hat{R}_{k\ell}\}) \right].&
            \end{array}
            $$
        \EndFor
    \State $\hat{S}, \{\hat{Q}_{k}\}_{k=1}^{K} \gets \operatorname{SecondStageRefit}\bigl(S^{(t)}, \{Q_{k}^{(t)}\}_{k=1}^{K}, \{\hat{R}_{k\ell}\}_{\mathcal{I}}; \{A_{k\ell}\}_{\mathcal{I}}\bigr)$

\end{algorithmic}
\end{algorithm}

\begin{proposition}\label{second_stage_gaus_case_proposition}
    With the edge entry distribution $f(\cdot; \theta, \sigma) = \mathcal{N}(\theta, \sigma^2)$, Problem \eqref{second-stage-conv-prob} coincides with the objective of the MultiNeSS model fitted to the layers $\tilde{A}_k = {1\over m_k}\sum_{\ell=1}^{m_k} (A_{k\ell} -\hat{R}_{k\ell}), \ k=1,\ldots, K$ under the assumption that their edge entries are independent and Gaussian with layer-dependent variances:
    $$\tilde{A}_{k, ij} \stackrel{\text{\textnormal{ind}}}{\sim} \mathcal{N}(S_{ij} +Q_{k, ij},  \sigma^2 / m_k), \quad 1\le i \le j \le n, \quad k=1, \ldots, K.$$
\end{proposition}

\subsection{Choice of tuning parameters}\label{hyperparam_tuning_section}

In this section, we propose a procedure for choosing the hyperparameters $\{\lambda_{1k}\}_{k=1}^K, \lambda_2,$ 
$\{\alpha_{2k}\}_{k=1}^K,$ and $\{\alpha_{1k\ell}\}_{\mathcal{I}}$ used in Algorithm \ref{alg:stepwise_update}, based on the commonly used edge cross-validation method of \cite{li_network_2020}. 

At first glance, the total number of tuning parameters seems overwhelming, but there are several ways to  significantly simplify cross-validation. Note that $\lambda_{1k}$ with $\{\alpha_{1k\ell}\}_{\ell=1}^{m_k}$ only appear in \eqref{first-stage-conv-prob} and $\lambda_{2}$ with $\{\alpha_{2k\ell}\}_{\ell=1}^{m_k}$ only in \eqref{second-stage-conv-prob}, which means we can tune these groups of parameters separately. We further fix $\alpha_{1k} := \alpha_{1k\ell}$ for $(k, \ell) \in \mathcal{I}$ and $\alpha_2 := \alpha_{2k}$ for $k=1, \ldots, K$, resulting in only two parameters to tune for each of the subproblems. As a simpler alternative, one could also just tune the more consequential $\lambda_2, \{\lambda_{1k}\}_{k=1}^K$ parameters, and set $\{\alpha_{1k\ell}\}_{\mathcal{I}}$ and $\{\alpha_{2k}\}_{k=1}^K$ to their theoretically optimal values derived for the Gaussian edge entry distribution in Corollary \ref{oracle_rate_corolalry}. In our simulations, this approach gave results comparable to tuning two hyperparameters per subproblem, and we set it as the default option in our implementation.

To tune the hyperparameters of the $k$-th subproblem in the first stage, we sample a random subset of ``training'' node pairs in layers of the $k$-th group:
$$\mathcal{A}^{(k)}_{train} \subset \mathcal{A}^{(k)} := \{A_{k\ell, ij}: \ell=1, \ldots, m_k, \ 1\le i\le j\le n\} . $$ 
We then solve \eqref{first-stage-conv-prob} with log-likelihood terms in $\mathcal{L}_k$ restricted to node pairs in $\mathcal{A}^{(k)}_{train}$ and then evaluate the non-penalized likelihood \eqref{group_log_likelihood} on the remaining ``test'' node pairs $\mathcal{A}_{test}^{(k)}  =\mathcal{A}^{(k)} \setminus \mathcal{A}^{(k)}_{train}$, and choose the parameters from a pre-defined grid to optimize the test log-likelihood averaged across multiple cross-validation folds for stability. Hyperparameters in the second stage problem \eqref{second-stage-conv-prob} are tuned similarly, with the only difference that node pairs are sampled from $A_{k\ell}, (k, \ell)\in\mathcal{I}$.


\subsection{Initialization}\label{init_section}
Both stages of the Algorithm \ref{alg:stepwise_update} require initial values, but they are convex problems, so the choice of initialization primarily impacts how long it takes the algorithm to converge, not its ability to attain the global optimum.  However, our theoretical results rely on the initializer being sufficiently close to the truth. 

 For the Gaussian edge-entry model, we initialize the shared component of a collection of layers as their average and the individual components as the resulting residuals.   When the ranks are known, these can be further truncated, resulting in initializers   
\begin{equation}\label{avg_init_first_stage}
    (S + Q_k)^{(0)} = \Bigl[ {1\over m_k}\sum_{\ell=1}^{m_k} A_{k\ell}\Bigr]_{d_0 + d_k}, \quad R_{k\ell}^{(0)} = \Bigl[A_{k\ell} - (S + Q_k)^{(0)}\Bigr]_{d_{k\ell}}, \quad \ell = 1, \ldots , m_k.
\end{equation}
For exponential family distributions with a non-linear link function $g$, we first get a proxy of $\Theta_{k\ell}$ by applying $g^{-1}$ to each layer $A_{k\ell}$ element-wise (truncating if necessary to avoid the inverse going to infinity), and then average the transformed layers. 

Alternatively, one can use the recently developed ``shared space hunting'' (SSH) approach (Algorithm 1 in \citet{TianYinqiu2024}).   This method is computationally more expensive than averaging as it requires first estimating the latent positions $X_{k\ell}$ of the nodes in each layer and the latent space dimensions $D_{k\ell}$.   Our simulations (available on GitHub) suggest that the SSH can produce very poor initializers if the latent dimensions are inaccurately estimated, which tends to happen when $n$ is small and/or the true ranks are large, whereas with correctly estimated ranks the SSH approach produces only marginally better results than averaging.  Thus we use averaging as the default initialization option, due to its robustness and low computational cost.


For Stage II, we have more initialization options to choose from since it can use the first-stage estimates $\{\widehat{S+Q_k}\}_{k=1}^K$ and $\{\hat{R}_{k\ell}\}_\mathcal{I}$. 
The natural candidates for the initializer of the shared component $S$ would be the average of either the estimates $\widehat{S + Q_k}$, 
\begin{equation}\label{mean_spq_init}
    S^{(0)}_{\text{sq}} = \Bigl[{1\over K} \sum_{k=1}^K \widehat{S + Q_k}\Bigr]_{d_0}
\end{equation}
or of the residuals between the adjacency matrix and the individual component estimate, 
\begin{equation}\label{second_stage_avg_resid_init}
    S^{(0)}_{\text{resid}} = \Bigl[{1\over K} \sum_{k=1}^K {1\over m_k} \sum_{\ell=1}^{m_k} (A_{k\ell} - \hat{R}_{k\ell})\Bigr]_{d_0}.
\end{equation}
Similarly to the first-stage initializer in \eqref{avg_init_first_stage}, for non-Gaussian distributions each layer $A_{k\ell}$ in \eqref{second_stage_avg_resid_init} should be first transformed by the inverse link function $g^{-1}$ and then truncated.
Alternatively, the SSH approach can be used for the initial estimation of the shared component $S$ from either $\{\widehat{S+Q}_k\}_{k=1}^{K}$ or
$\Bigl\{{1\over m_k}\sum_{\ell=1}^{m_k}(A_{k\ell} - \hat{R}_{k\ell})\Bigr\}_{k=1}^K$.

In our implementation of Algorithm \ref{alg:stepwise_update}, we used \eqref{second_stage_avg_resid_init} as the second-stage initializer, as it gives the most direct parallel to MultiNeSS per Proposition \ref{second_stage_gaus_case_proposition}. 
However, empirically we observed that both initializations \eqref{second_stage_avg_resid_init} and \eqref{mean_spq_init} result in essentially the same convergence speeds for problem \eqref{second-stage-conv-prob}. 
To provide a possible theoretical explanation to this observation,  we derive an explicit relationship between the two under the Gaussian edge distribution and a mild extra assumption on the self-loop distribution.   The following proposition demonstrates that in the Gaussian case, $\widehat{S + Q_k}$ is the soft-thresholded version of the averaged residuals ${1\over m_k} \sum_{\ell=1}^{m_k} \bigl(A_{k\ell} - \hat{R}_{k\ell}\bigr)$.  The proof can be found in Section \ref{fit_section_proofs} of the Appendix.

\begin{proposition}\label{SpQ_avg_residual_relationship_propos}
Assume that the edge entry distribution is $f(\cdot;\theta, \sigma) = \mathcal{N}(\theta, \sigma^2)$ for the layers' off-diagonal entries and $\mathcal{N}(\theta, 2\sigma^2)$ for the diagonal entries. Then, the estimates produced by solving problem \eqref{first-stage-conv-prob} are related as follows:
\begin{equation}\label{SpQ_avg_residual_relationship}
       \widehat{S+Q}_k = \mathcal{T}_{2\sigma^2\lambda_{1k} / m_k}\Bigl[{1\over m_k}\sum_{\ell=1}^{m_k}\bigl(A_{k\ell} - \hat{R}_{k\ell}\bigr)\Bigr].
   \end{equation}
\end{proposition}

\section{Consistency results}
    \label{sec:theory}

This section establishes theoretical guarantees for the parameter estimates obtained by running Algorithm \ref{alg:stepwise_update} with unit learning rates $\eta_1=\eta_2=1$ under the Gaussian assumption on edge distribution, $f(\cdot; \theta, \sigma)=\mathcal{N}(\theta, \sigma^2)$. To simplify our analysis, we assume that Stage II and each subproblem of Stage I are terminated after the first iteration $(t=1)$. 
Throughout this section, we assume that $M$, $K$, $m_k$, $d_0$, $d_k$, and $d_{k\ell}$  are possibly increasing functions of $n$. The variance $\sigma^2$ is assumed fixed. All proofs for this section can be found in Section \ref{consist_section_proofs} of the Appendix.

We begin by introducing additional notation. 
For $a, b\in\mathbb{R}$, denote $a\vee b = \max(a, b)$. For real-valued sequences $g_n, h_n$, we write  $g_n \lesssim h_n$ if $g_n = O(h_n)$ and $g_n \asymp h_n$ if  $g_n \lesssim h_n$ and $h_n \lesssim g_n$. For brevity, we write  
$$ \rho_{1k} := \lambda_{1k} / m_k, \quad \rho_{1k\ell} := \lambda_{1k} \alpha_{1k\ell}, \quad \rho_2 := \lambda_2 / M,\quad \rho_{2k} := \lambda_2\alpha_{2k} / m_k, \quad \quad \text{for} \quad (k, \ell) \in \mathcal{I}.$$
Note there is a one-to-one correspondence between these newly defined thresholds and the original  parameters  $\{\lambda_{1k}\}_{k=1}^K, \lambda_2, \{\alpha_{2k}\}_{k=1}^K, \{\alpha_{1k\ell}\}_{\mathcal{I}}$.
For matrices $Z_1, Z_2 \in \mathbb{R}^{n\times n}$ we denote the cosine similarity between their eigenspaces by 
$$\cos(Z_1, Z_2) = \max_{x\in\operatorname{col} (Z_1), \ y\in \operatorname{col}(Z_2)} {|x^\top y| \over \|x\|_2\|y\|_2}.
$$
For convenience, we also define maximal angles between the following pairs of latent component types
\begin{equation} \label{cos_sim_comp_pairs}
    \begin{aligned}
    &s_{v,w}  = \max_{1\le k \le K} \cos(S, Q_k),  &s_{w,w}  & = \max_{1\le k_1 < k_2 \le K} \cos(Q_{k_1}, Q_{k_2}),  \nonumber  \\ 
  &s_{vw,u}^{(k)} =  \max_{1\le \ell\le m_k} \cos(S +Q_k, R_{k\ell}),    
    &s_{u,u}^{(k)} & =  \max_{1\le \ell_1< \ell_2 \le m_k} \cos(R_{k\ell_1}, R_{k\ell_2}),\\
    &s_{u,u} =  \max_{(k_1, \ell_1)\ne (k_2, \ell_2)} \cos(R_{k_1\ell_1}, R_{k_2\ell_2}). & & 
    \end{aligned}
\end{equation}
These quantities will appear in the error bounds on the latent components, supporting the intuition that the accuracy of separating latent spaces depends on how similar they are to each other.  

Before we proceed to the main result, we state a concentration bound that depends on the Gaussian assumption. With a minor modification of the algorithm, the main result can be extended to sub-Gaussian distributions, as discussed below in Remark \ref{gaus_relaxation_remark}.   Rewriting the observed matrix as signal plus noise, we have 
\begin{equation}\label{layer_decomp_with_gaus_error}
    A_{k\ell} = \Theta_{k\ell} + E_{k\ell} = (S + Q_k + R_{k\ell}) + E_{k\ell}, 
\end{equation}  
 where $E_{k\ell} \in \mathbb{R}^{n \times n}$ is a symmetric centered noise matrix with $E_{k\ell, ij} \stackrel{\text{ind}}{\sim} \mathcal{N}(0,\sigma^2)$ for $1\le i\le j\le n$ and $(k, \ell)\in\mathcal{I}$,  
Gaussian matrices enjoy a convenient concentration bound on their operator norm, as well as the averages of their independent copies. In particular, our theoretical analysis will be restricted to the event where all individual errors $E_{k\ell}, (k, \ell)\in \mathcal{I}$, their groupwise averages $\bar{E}_k = {1\over m_k}\sum_{\ell=1}^{m_k} E_{k\ell}, k=1, \ldots, K$, and total average $\bar{E} = {1\over M}\sum_{(k,\ell)\in\mathcal{I}}E_{k\ell}$ are bounded as follows:
\begin{equation}\label{error_bound_set}
    \mathcal{E}_{noise} := \left\{\|E_{k\ell}\|_2 \le 3\sigma\sqrt{n}, \ (k, \ell) \in\mathcal{I}; \  \|\bar{E}_k\|_2 \le 3\sigma \sqrt{n / m_k}, \ k=1,\ldots, K;\ \|\bar{E}\|_2 \le 3\sigma \sqrt{n / M}\right\}
\end{equation}
The following lemma shows that this event has a high probability. The proof is analogous to Lemma 3 in \citep{multiness}.
\begin{lemma}\label{gaus_noise_lemma}
With a universal constant $C_0 > 0$, it holds $\mathbb{P}\bigl(\mathcal{E}_{noise}\bigr) \ge 1 - (M + K + 1)ne^{-C_0n}$.
\end{lemma}

Next, we state assumptions needed to establish consistency. The first says that all groups have comparable sizes, asymptotically dominating the number of groups.
\begin{assumption}\label{group_props_assump}
    For each group $k=1, \ldots, K$, $m_k \asymp M / K $ with  $K \lesssim M^{1/2}$ and $M\rightarrow\infty$.
\end{assumption}

The next assumption controls the signal-to-noise ratio, calibrated relative to the concentration bound in Lemma \ref{gaus_noise_lemma}.
\begin{assumption}\label{eigenvals_order_gmn} 
There are constants $0< b_S < B_S$, $0< b_Q < B_Q$, $0 < b_{S+Q} < B_{S+Q}$, $0< b_R < B_R$, and $\tau \in (1/2, 1]$, such that
\begin{equation}
    \begin{aligned}
& b_S n^\tau \leq\left|\gamma_{d_0}\left(S\right)\right| \leq\left|\gamma_1\left(S\right)\right| \leq B_S n^\tau, & \\
& b_Q n^\tau \leq\left|\gamma_{d_k}\left(Q_k\right)\right|\leq\left|\gamma_1\left(Q_k\right)\right| \leq B_Q n^\tau,  &  k=1, \ldots, K,\\
& b_{S +Q} n^\tau \leq\left|\gamma_{d_0 + d_k}\left(S + Q_k\right)\right|\leq\left|\gamma_1\left(S+Q_k\right)\right| \leq B_{S+Q} n^\tau,  &  k=1, \ldots, K,\\
& b_R n^\tau \leq\left|\gamma_{d_{k\ell}}(R_{k\ell})\right| \leq\left|\gamma_1(R_{k\ell})\right|\leq B_R n^\tau,  & (k, \ell)\in\mathcal{I}
\end{aligned}
\end{equation}
\end{assumption} 
Since we assumed that both Stage I and Stage II of Algorithm \ref{alg:stepwise_update} are run for one iteration only, it is important to require that the initializers for each stage are "good enough".  Rather than state the result for any particular choice of initializers, we formulate this requirement as a separate generic assumption.  We will later show that our choice of initializer satisfies this assumption. 

\begin{assumption}\label{initializer_assump}
     For $n$ sufficiently large, Stage I and II initializers on the set $\mathcal{E}_{noise}$ satisfy 
    \begin{equation}\label{init_error_set}
   \mathbb{P}   \Bigl\{\|S + Q_k - (S + Q_k)^{(0)}\|_2 \le r_k^{(I)},  
   \ \|S - S^{(0)}\|_2 \le r^{(II)} \, | \, \mathcal{E}_{noise}  \} = 1,
    \end{equation}
 where $r^{(I)}_k$ and $r^{(II)}$ are some $o(n^\tau)$ deterministic functions of $n$.

 
\end{assumption}
In Proposition \ref{init_sqrt_n_proposition}, we state explicit conditions for the initializers \eqref{avg_init_first_stage}  and \eqref{second_stage_avg_resid_init} to satisfy this assumption with errors of order $\sqrt{n}$. Alternatively, Theorem 1 of \citet{TianYinqiu2024} provides different conditions under which the Frobenius norm (and thus the operator norm) of the SSH initializer's error is of order $\sqrt{n}\log n$. 

With Assumptions \ref{group_props_assump}, \ref{eigenvals_order_gmn}, and \ref{initializer_assump} at hand, we are ready to state the main result.

\begin{theorem} \label{main_consistency_theorem} Suppose the edge-entry distribution is $f(\cdot ; \theta, \sigma)=\mathcal{N}\left(\theta, \sigma^2\right)$ with fixed $\sigma^2$. 
Then under Assumptions \ref{group_props_assump}, \ref{eigenvals_order_gmn}, \ref{initializer_assump}, with probability greater than $1-(M+K+1) n e^{-C_0 n}$, where $C_0 > 0$ is a universal constant, and for $n$ sufficiently large, the estimates produced by Algorithm \ref{alg:stepwise_update} with learning rates $\eta_1=\eta_2=1$ and without refitting satisfy, for $(k, \ell) \in \mathcal{I}$, 
\begin{equation*}
    \begin{array}{lll}
    \textnormal{(Stage I)} &\|\hat{R}_{k\ell} - R_{k\ell}\|_F \leq 4 d_{k\ell}^{1 / 2}\rho_{1k\ell}, &  \|\widehat{S + Q}_{k} - (S + Q_k)\|_F \leq 4 (d_0 + d_{k})^{1 / 2}\rho_{1k}\\
    \textnormal{(Stage II)} &\|\hat{Q}_k-Q_k\|_F \leq 4 d_k^{1 / 2}\rho_{2k}, & \|\hat{S}-S\|_F \leq 4 d_0^{1 / 2}\rho_2,
    \end{array}
\end{equation*}
where 
\begin{equation*}
\begin{aligned}
    &\rho_{1k\ell} \asymp r_k^{(I)}\vee n^{1 / 2} ,  \\
    &\rho_{1k}   \asymp \max_{1\le \ell\le m_k}\rho_{1k\ell}\Bigl[{1\over m_k} \vee s_{u, u}^{(k)} \vee n^{-\tau}\max_{1\le \ell\le m_k}\rho_{1k\ell}\Bigr]^{1/2},  \\
     &\rho_{2k} \asymp r^{(II)} \vee \rho_{1k},  \\
     &\rho_2  \asymp \max_{1\le k\le K}\rho_{2k}\Bigl[{1\over K} \vee s_{w, w}  \vee n^{-\tau} \max_{1\le k\le K}\rho_{2k}\Bigr]^{1/2}
    \vee \max_{(k,\ell)\in\mathcal{I}}\rho_{1k\ell}& \Bigl[{1\over M} \vee s_{u, u} \vee n^{-\tau} \max_{(k,\ell)\in\mathcal{I}}\rho_{1k\ell}\Bigr]^{1/2}, 
\end{aligned}
\end{equation*}
and the constants in the rates depend only on  $(B_R, b_R, B_{S + Q}, b_{S+Q}, B_Q, b_Q, B_S, b_S, \sigma)$.  
Additionally, if $S$, $Q_k,$ and $R_{k\ell}$ are positive semi-definite (PSD), then $\hat{S}, \hat{Q}_k, \hat{R}_{k\ell}$ are also PSD.
\end{theorem}

\begin{remark}
An important advantage of Algorithm \ref{alg:stepwise_update} is the convexity of the optimization problems, which should mean the initial values do not affect consistency, only the speed of convergence.  We made an assumption on the initializer error in Theorem \ref{main_consistency_theorem} because we only analyze one update of the gradient descent, due to the complicated alternating updates.  We leave it for future work to formally remove the initialization error assumption, though by convexity it is clear that even if the chosen initial value does not satisfy it (and we show that ours do), after some number of gradient descent steps it will. 
\end{remark}

\begin{remark}\label{gaus_relaxation_remark}
    Our proof of Theorem \ref{main_consistency_theorem} relies on two key properties of the Gaussian distribution: (i) the convenient form of the proximal gradient updates for the Gaussian log-likelihood, and (ii) the concentration result of \citet{Bandeira2016} used to establish Lemma \ref{gaus_noise_lemma}.  The concentration result was extended to sub-Gaussian distributions in the same paper (Corollary 3.3).  The convenient form of gradient updates can be retained by replacing the log-likelihood with the sum of squares loss. This implies Theorem \ref{main_consistency_theorem} can be extended to sub-Gaussian edge distributions, and in particular to the RDPG binary edge model with $\Theta_{k\ell} \in [0, 1]^{n\times n}$ and $ A_{k\ell}\sim \operatorname{Bernoulli}(\Theta_{k\ell}), (k, \ell)\in \mathcal{I}$. 
\end{remark}

Theorem \ref{main_consistency_theorem} expresses the rates in terms of key quantities such as
$r_k^{(I)}$, $r^{(II)}$, $m_k$, and $n$, allowing us to apply the analysis to a wide range of scenarios. For instance, we can derive the conditions under which Algorithm \ref{alg:stepwise_update} produces estimates matching the oracle error rates,  where each parameter $S, \{Q_k\}_{k=1}^K, \{R_{k\ell}\}_\mathcal{I}$ is estimated using its true rank and true values of all other parameters:
\begin{equation}
    \begin{aligned}
        &\hat{S}^{(oracle)}=\Bigl[\frac{1}{M} \sum_{ \mathcal{I}}\left(A_{k\ell} -Q_k -  R_{k\ell}\right)\Bigr]_{d_0} =  \bigl[S + \bar{E}\bigr]_{d_0}\ ,  \\ 
        &\hat{Q}_k^{(oracle)}=\Bigl[\frac{1}{m_k} \sum_{\ell=1}^{m_k}\left(A_{k\ell} -S - R_{k\ell}\right)\Bigr]_{d_k} =  \bigl[Q_k + \bar{E}_k\bigr]_{d_k}\ , \\ 
        & \hat{R}_{k\ell}^{(oracle)} = 
\left[A_{k\ell}- S - Q_k\right]_{d_{k\ell}}  = \bigl[R_{k\ell} + E_{k\ell}\bigr]_{d_{k\ell}}\ .
    \end{aligned}
\end{equation}
  where the second set of equalities follows from \eqref{layer_decomp_with_gaus_error}. Per Lemma \ref{hard_threshold_lemma}, this shows that their errors are dominated by the rates for the corresponding noise components in \eqref{error_bound_set}.
We state sufficient conditions ensuring that Algorithm \ref{alg:stepwise_update} achieves the oracle rates in the following corollary of Theorem \ref{main_consistency_theorem}. We begin by stating an additional assumption on the rates of group sizes and pairwise similarities of individual and group components.

\begin{assumption}\label{rates_assump}
It holds $s_{u, u} \lesssim n^{1/2 - \tau}, \  s_{w,w}\lesssim K^{-1},$ and $m_k\lesssim n^{\tau -1/2}$ for each $k=1,\ldots, K$.
\end{assumption}


\begin{corollary}[Oracle rate conditions]
\label{oracle_rate_corolalry}
    Under Assumption \ref{rates_assump} and assumptions of Theorem \ref{main_consistency_theorem} with $r^{(II)} \lesssim (nK/M)^{1/2}$ and $r_k^{(I)}\lesssim n^{1/2}$ for each $k=1,\ldots, K$, it holds
    
    \begin{equation*}
        \begin{array}{lll}
    \|\hat{R}_{k\ell} - R_{k\ell}\|_F \leq C_R d_{k\ell}^{1 / 2}n^{1/2}, &  \|\widehat{S + Q}_{k} - (S + Q_k)\|_F \leq C_{S+Q} (d_0 + d_{k})^{1 / 2}n^{1/2}m_k^{-1/2}\\
    \|\hat{Q}_k-Q_k\|_F \leq C_Qd_k^{1 / 2}n^{1/2}m_k^{-1/2}, & \|\hat{S}-S\|_F \leq C_S d_0^{1/2}n^{1/2} M^{-1/2},
    \end{array}
\end{equation*}
    if the hyperparameters are set with sufficiently large positive constants $c_{1k}$ and $c_2$ as
    \begin{equation}\label{oracle_hyperparam_rates}
        \lambda_{1k} = c_{1k}\sqrt{nm_k}, \quad \alpha_{1k\ell} = 1/\sqrt{m_k}, \quad \lambda_{2} = c_2\sqrt{nM}, \quad \alpha_{2k} = \sqrt{m_k / M}, \quad \quad (k, \ell)\in\mathcal{I}.
    \end{equation}

\end{corollary}

We conclude this section by stating sufficient conditions, which guarantee that our initializers of $\{S + Q_k\}_{k=1}^K$ in the first stage and of $S$ in the second stage have the rates as in Corollary \eqref{oracle_rate_corolalry}.

\begin{proposition}\label{init_sqrt_n_proposition} Under Assumptions \ref{group_props_assump}, \ref{eigenvals_order_gmn}, \ref{rates_assump}, if $s_{vw, u}^{(k)} \lesssim n^{1/2-\tau}m_k^{1/2}, s_{v,w}\lesssim Kn^{1/2-\tau} / M^{1/2}$, and there is sufficient separation between the spectra of $S$ and $Q_k$'s, that is, there exists $\delta > 0$ such that
    \begin{equation}\label{min_S_max_Q_sep}
             \frac{|\gamma_{d_S}(S)|}{\max_{k}|\gamma_{1}(Q_k)|} \ge {b_S \over B_Q} \ge 4(1 + \delta)\Bigl[{1\over K} + s_{w,w}\Bigr]^{1/2},
    \end{equation}
     for $n$ sufficiently large, Assumption \ref{initializer_assump} is satisfied for the initializers in \eqref{avg_init_first_stage} and \eqref{second_stage_avg_resid_init} with   $r^{(II)} \asymp (nK/M)^{1/2}$ and $r^{(I)}_k \asymp n^{1/2}, \ k=1, \ldots, K$.
\end{proposition}

\section{Synthetic networks experiments}
	\label{sec:experiments}
	In this section, we empirically study the properties of Algorithm \ref{alg:stepwise_update} in various synthetic scenarios. In Section \ref{param_dependency_exp_section}, we investigate how the accuracy of  GroupMultiNeSS is affected by key quantities such as the size of the networks, the number of layers, and the similarities between latent components. In Section \ref{compar_other_methods_exp_section}, we compare GroupMultiNeSS to other models for multiplex networks and demonstrate its advantages over existing methods when a latent group structure is present in the layers.

\subsection{Experimental settings}
We let all ranks $\{d_k\}_{k=0}^K$, $\{d_{k\ell}\}_{\mathcal{I}}$ of all latent components be the same and denote them by $d$. The embedding similarity measure $\kappa$ is taken to be the standard Euclidean inner product. 
To generate latent components with varying pairwise maximum cosine similarities, we use Algorithm \ref{alg:sampling_procedure} below, inspired by a similar sampling approach of \cite{TianYinqiu2024}. The algorithm uses two additional notions of maximal angles: 
$$s_{v, u} =  \max_{(k, \ell) \in \mathcal{I}} \cos(S, R_{k\ell}), \qquad s_{w, u}=  \max_{k=1,\ldots, K}\max_{1\le \ell\le m_k} \cos(Q_k, R_{k\ell}).$$  

The proposed sampling approach ensures that the columns with distinct indices within any two latent components are orthogonal, and columns with identical indices have similarity depending only on the types (shared, group, or individual) of the two input components. In particular, this implies that the sampling procedure is valid only if $d(1 + K + M) \le n$, as otherwise the number of angle constraints is larger than the number of available degrees of freedom $n$.

\begin{algorithm}
\caption{Latent component sampling algorithm \label{alg:sampling_procedure}}
\begin{algorithmic}[1]
    \State \textbf{Input:} number of nodes $n$, latent dimension $d$, group sizes $\{m_k\}_{k=1}^K$, maximum cosine angles $s_{v, w}, s_{v, u}, s_{w, w}, s_{w, u}, s_{u, u}$ (all zero by default)
    \State Collect all angles into a single matrix:
    \begin{equation}\label{component_angle_matrix}
    \Omega = \begin{pmatrix}
        1 & s_{v, w}\mathbf{1}_K^\top & s_{v, u}\mathbf{1}_M^\top \\
         s_{v, w}\mathbf{1}_K & \Sigma_K(s_{w, w}) & s_{w, u} \mathbf{1}_K\mathbf{1}_M^\top \\
        s_{v, u}\mathbf{1}_M & s_{w, u}\mathbf{1}_M\mathbf{1}_K^\top & \Sigma_M(s_{u, u})
    \end{pmatrix}
    \end{equation}
    where $\Sigma_m(s)$ is an $m\times m$ matrix with ones on the diagonal and $s$ everywhere else. 
    \State Initialize latent positions as i.i.d.\ draws from the standard normal distribution
    $$\tilde{L} = [\tilde{V}, \tilde{W}_1, \ldots, \tilde{W}_K, \tilde{U}_{11}, \ldots, \tilde{U}_{1m_1}, \ldots, \tilde{U}_{K1}, \ldots, \tilde{U}_{Km_K}] \in \mathbb{R}^{n \times d(1 + K + M)}
    $$
 
    \State Compute the initial and target Gram matrices as, respectively, $\tilde{\mathcal{G}} = {1\over n} \tilde{L}^\top\tilde{L}$ and $\mathcal{G} = \Omega \otimes I_d$ and set the latent positions to 
    $$\tilde{L}\tilde{\mathcal{G}}^{-1/2}\mathcal{G}^{1/2} =  [V, W_1, \ldots, W_K, U_{11}, \ldots, U_{1m_1}, \ldots, U_{K1}, \ldots, U_{Km_K}]$$
    \State \textbf{Output:} $S=VV^\top, Q_k = W_kW_k^\top, R_{k\ell} = U_{k\ell}U_{k\ell}^\top$ for all $(k, \ell)\in\mathcal{I}$. 
\end{algorithmic}
\end{algorithm}

When fitting the GroupMultiNeSS model with Algorithm \ref{alg:stepwise_update}, we use a 5-fold edge cross-validation to tune $\{\lambda_{1k}\}_{k=1}^K$ and $\lambda_2$ as described in Section \ref{hyperparam_tuning_section}, with training and test sets within each fold of size $|\mathcal{A}_{train}| = 0.8|\mathcal{A}|$ and $|\mathcal{A}_{test}| = 0.2|\mathcal{A}|$. For hyperparameter tuning, we define the grid $c_\lambda\in [0.03, 0.1, 0.3, 1, 3, 10]$ and use the oracle rates in \eqref{oracle_hyperparam_rates}, that is, we set $\lambda_{1k} = c_\lambda \sqrt{nm_k}$ and $\lambda_{2} = c_\lambda\sqrt{nM}$. In all experiments, we use the oracle values for $\{\alpha_{1k\ell}\}_{\mathcal{I}}$ and $\{\alpha_{2k}\}_{k=1}^K$. 
We measure convergence of the algorithm by the relative difference between the current loss and its best value achieved so far, 
and stop if this difference has not exceeded the tolerance value $10^{-5}$ during the last ten steps. The learning rates used in Algorithm \ref{alg:stepwise_update} are set to $\eta_1=\eta_2=1$ for the Gaussian model and to $\eta_1=\eta_2=3$ for the Bernoulli edges.

For any parameter matrix $\Theta$, we measure the error of its estimate $\hat\Theta$ using the Relative Frobenius Error (RFE), defined  as  $ \operatorname{RFE}(\hat{\Theta}, \Theta) := {\|\hat{\Theta} - \Theta\|_F /  \|\Theta \|_F}$.   For a collection of matrices $\{\Theta_\ell\}_{\ell=1}^m$ and their corresponding estimates $\{\hat{\Theta}_\ell\}_{\ell=1}^m$, we sometimes report the average RFE (ARFE),  defined as:
${m^{-1}}\sum_{\ell=1}^{m} \operatorname{RFE}(\hat{\Theta}_\ell, \Theta_\ell) $.
We use this metric to measure the average estimation accuracy for a given type of latent components. 

Python implementations of GroupMultiNeSS (Algorithm \ref{alg:stepwise_update}), MultiNeSS \citep{multiness}, and Shared Space Hunting with refinement \citep{TianYinqiu2024}, as well as the synthetic data sampler (Algorithm \ref{alg:sampling_procedure}), are available in the \texttt{GroupMultiNeSS} package at 
\url{https://github.com/AlexanderKagan/GroupMultiNeSS}.  All simulation studies can be found at 
\url{https://github.com/AlexanderKagan/GroupmultinessExperiments}.

\subsection{Accuracy of GroupMultiNeSS}\label{param_dependency_exp_section}

In this section, we study how estimation accuracy of the GroupMultiNeSS latent components $S, \{Q_k\}_{k=1}^K, \{R_{k\ell}\}_{\mathcal{I}}$ and the resulting parameter matrices $\Theta_{k\ell} = S + Q_k + R_{k\ell}$ depends on the number of nodes $n$ and the number of layers $M$. 

We generate latent components using Algorithm \ref{alg:sampling_procedure} with $K=4$ balanced groups, $s_{v, u} = s_{w, u} = 0.1$, and latent dimensions $d=3$. The observed layers are sampled as  
\begin{equation}\label{layer_generation_experiment} 
A_{k\ell, ij} \stackrel{\text{ind}}{\sim} \mathcal{N}(\Theta_{k\ell, ij}, 1) \quad \text{or} \quad A_{k\ell, ij} \stackrel{\text{ind}}{\sim} \operatorname{Bernoulli}\bigl(\sigma(\Theta_{k\ell, ij})\bigr),\qquad (k, \ell) \in\mathcal{I}, \quad i\le j,
\end{equation}
where $\sigma$ is the logistic link function.  
Figure \ref{fig:parameter_dependency} presents the estimation errors as a function of the number of layers $M\in [8, 16, 24, 32, 40, 48]$ with $n=200$ and the number of nodes $n\in[100, 200, 300, 400, 500]$ with $M=16$ for the Gaussian and logistic models. All  experiments were repeated ten times with different random seeds to assess empirical standard errors, and showed Monte Carlo errors to be negligible, so error bars are omitted in the plots. 

As expected from our theoretical analysis, the estimation error of all components decreases as the number of nodes $n$ increases, since the true matrices are all low rank.  Increasing the number of layers $M$ improves the estimation of parameters that are shared by multiple layers --  $S$, $Q$, and therefore $\Theta$ -- but does not much affect the error in estimating the individual component $R$, which matches both the bounds of Theorem \ref{main_consistency_theorem} and  the intuition other layers do not help estimate the individual component of a given layer.    The slight improvement in $R$ shown in the plots is likely due to improved estimates of the shared and group components, which indirectly helps separate $R_{k\ell}$ from $S+Q_k$.

Comparing logistic and Gaussian models, we observe that the relative error of $\Theta$ is lower than that of all other components in the logistic case and lies between $R$ and $(Q, S)$ in the Gaussian case.  We explain it by the fact that separating additive latent components is a much easier task under a linear link function than a logistic one.  This also explains the much higher relative errors of all components in the logistic case.

\begin{figure}[t!]
    \centering
    \includegraphics[width=0.9\linewidth]{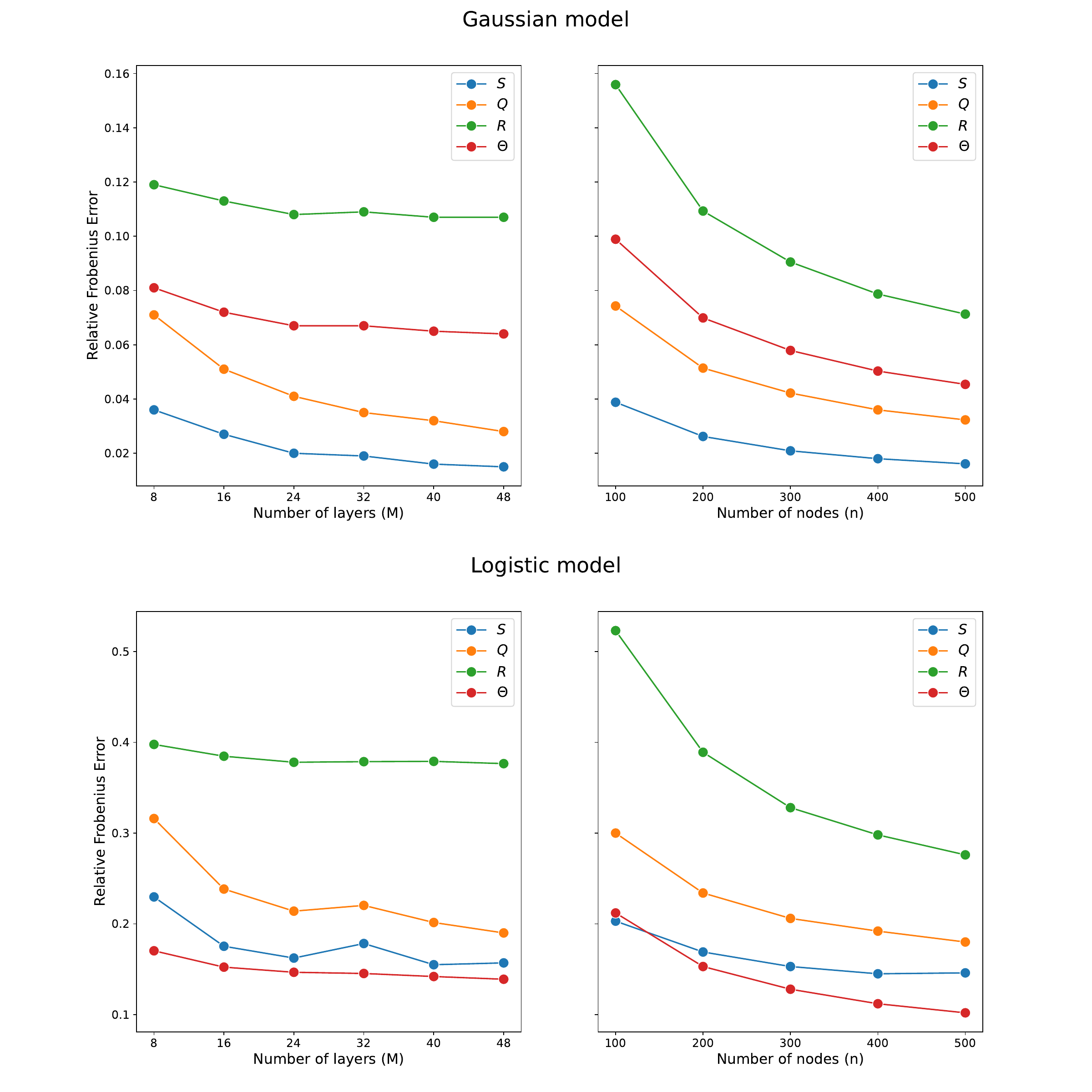}
    \caption{ ARFE for the Gaussian (top row) and logistic (bottom row) GroupMultiNeSS models, for $K=4$ balanced groups, $d=3$ latent dimensions, and $s_{v, w}=s_{w, u} = 0.1$.  Left column:  $n = 200$, $M$ varies from 8 to 48.   Right column:  $M = 16$, $n$ varies from 100 to 500. }
    \label{fig:parameter_dependency}
\end{figure}

\subsection{Comparison to other methods}\label{compar_other_methods_exp_section}

Here, we compare Algorithm \ref{alg:stepwise_update} with several other methods for multiplex networks with shared structure.   As a baseline, we include the MultiNeSS model \citep{multiness},  with the refitting step and nuclear norm penalty hyperparameters chosen by cross-validation.    We also include 
the Multiple Adjacency Spectral Embedding (MASE) algorithm, proposed by \cite{Arroyo2019} for fitting their  COSIE multiplex network model. COSIE  estimates the expectations of each layer but does not separate shared and individual components,  thus we can only compare the accuracy of estimating the overall expectation $\Theta_{k\ell}$.  Although COSIE is designed for the RDPG model, it can be directly applied to the Gaussian model. We implement an oracle version of COSIE that uses true ranks $\{d_k\}_{k=0}^K$ and $\{d_{k\ell}\}_{\mathcal{I}}$. To ensure a fair comparison with our method, we select the leading $d_0 + d_k + d_{k\ell}, \ (k, \ell)\in\mathcal{I}$ eigenvectors for each layer, and then use COSIE to fit a common invariant subspace of dimension $d_{0} + \sum_{k=1}^K d_{k} + \sum_{(k, \ell)\in \mathcal{I}} d_{k\ell}$, corresponding to the total number of latent dimensions in the GroupMultiNeSS model.
Finally, as a gold-standard benchmark, we include the oracle version of GroupMultiNeSS, which replaces soft thresholding updates by hard thresholding at true ranks, as stated in \eqref{proximal_map_operator_hard}, and omits the the refitting step since hard thresholding does not shrink eigenvalues. 

\begin{figure}[ht!]
    \centering
    \includegraphics[width=0.97\linewidth]{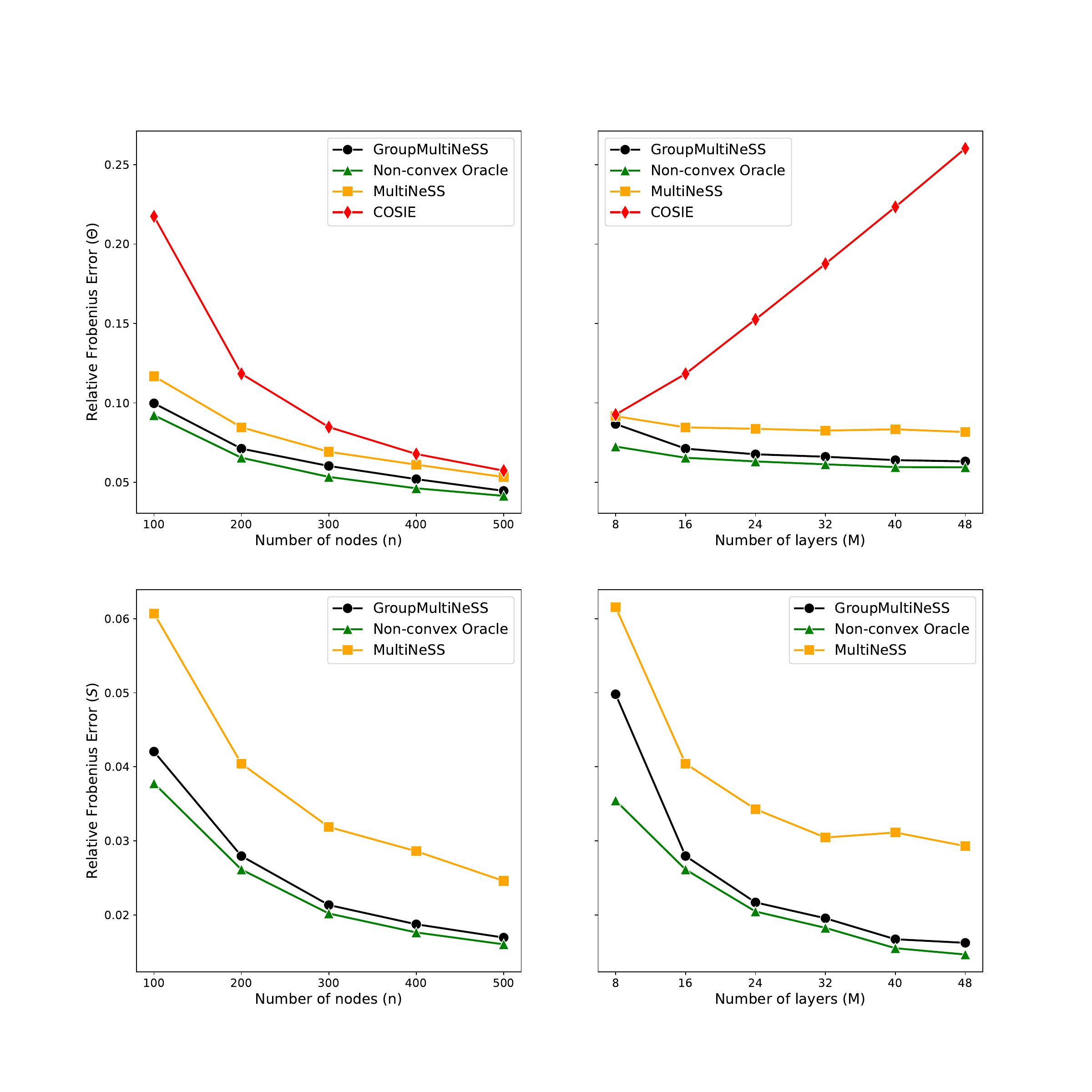}
    \label{fig:errors_across_models_gaus}
    \caption{ ARFE of $\Theta$ (top row) and $S$ (bottom row) as a function of the number of nodes $n$ with $M = 16$ (left column) and the number of layers $M$ with $n=200$ (right column), for $K=4$ balanced groups, latent dimension $d=3$, and the Gaussian edge distribution.}
\end{figure}

We keep the setting the same as in the previous section, with $K=4$ balanced groups and latent dimension $d=3$, generating the networks from the Gaussian distribution, with  $n\in [100, 200, 300, 400, 500]$ and  $M\in [8, 16, 24, 32, 40, 48]$.    The ARFE of the expectation matrices $\Theta$ for the four methods (GroupMultiNeSS, Oracle GroupMultiNeSS, MultiNeSS, and COSIE) in the top row of Figure \ref{fig:errors_across_models_gaus}.   The bottom row also shows the RFE of the extracted shared component $S$ for the first three methods.   Overall, the GroupMultiNeSS performance comes close to its oracle version and outperforms the other methods.    MultiNeSS does a reasonable job on estimating the overall mean $\Theta$, with error about 20\% higher than GroupMultiNeSS, but has about 50\% higher error on the shared component $S$;  this agrees with the results in  Figure \ref{fig:parameter_dependency}, showing that the ``total'' latent space $\Theta$ can be estimated relatively well even when its components are not accurately separated from each other.

While all methods improve as $n$ groups and both MultiNeSS and GroupMultiNeSS improve with $M$ increasing as well, COSIE gets worse with $M$;  this is a well-known property of the MASE algorithm, since it involves estimating the joint latent space of dimension that becomes more severely misspecified as $M$ increases. Results for the logistic model are similar and can be found in Section \ref{logistic_compar_appendix_section} of the Appendix.

\section{Application to Parkinson's disease brain networks}
    \label{sec:real_data}
	To demonstrate the practical utility of the proposed GroupMultiNeSS model, we analyze a publicly available functional connectivity dataset curated by \cite{TaoWu2017}. It consists of resting-state fMRI data from 40 subjects, 17 women and 23 men aged between 57 and 75 years, among whom 20 have been diagnosed with Parkinson’s disease (PD) and 20 are healthy controls. The functional brain network of each subject is represented by a Pearson correlation matrix computed across the 116 brain regions (nodes) defined by the AAL116 atlas \citep{AAL116_atlas}.   This atlas divides these $n = 116$ regions into 8 brain systems: there are 28 nodes in the frontal lobe, 26 in the cerebellum, 14 in the occipital lobe, 14 in the parietal lobe, 12 in the limbic system, 12 in the temporal lobe, 8 in subcortical gray matter (SCGM), and 2 in the insula. 

Our goal is to use the GroupMultiNeSS model to identify differences in the latent structure between the PD patients and controls.  Following the standard processing pipeline in neuroimaging, we apply the Fisher $z$-transformation to Pearson correlations, making the Gaussian edge model suitable for this application.  The diagonal elements of the correlation matrices are omitted from the loss function in optimization.  
We also apply the standard preprocessing step of controlling for age and sex by regressing them out of each edge entry (separately for each edge).   That is, we replace each Fisher-transformed correlation with the residual from regressing all Fisher-transformed correlations for that pair of nodes on the subjects' sex and age.  


We fit the model by Algorithm \ref{alg:stepwise_update} using the generalized inner product as the similarity function 
and used edge cross-validation to choose the hyperparameters $\{\lambda_{1k}\}_{k=1}^K$ and $\lambda_{2}$.  After fitting the latent components $\hat{S}, \{\hat{Q}_k\}_{k=1}^K,$ and $\{\hat{R}_{k\ell}\}_{\mathcal{I}}$, we extracted the latent positions $\hat{V}, \{\hat{W}_k\}_{k=1}^K,$ and $\{\hat{U}_{k\ell}\}_{\mathcal{I}}$ using ASE as described at the beginning of Section \ref{sec:estimation}.
In Figure \ref{fig:gmn_group_components_taowu}, we plot estimated group-specific latent positions in the three leading latent dimensions of $\{\hat{W}_k\}_{k=1}^K$ with regions (nodes) colored according to the brain system they belong to. For better visualization, we align the embeddings of the two groups by applying the three-dimensional  rotation obtained by Procrustes alignment, and only include the five biggest systems (frontal, occipital, parietal, temporal, and cerebellum), which together include 94 nodes.  
The three leading dimensions for both groups are estimated to be disassortative, which roughly means that a larger inner product in this dimension is associated with a smaller edge weight. 
Together, the three dimensions explain roughly $43\%$ and $47\%$ of the variance in the control and PD groups, respectively, as measured by the sum of all singular values.   


While we do not know the ground truth in real data, comparing the two group embeddings in Figure \ref{fig:gmn_group_components_taowu}, we can observe large differences in the cerebellum (purple)  and occipital (orange) systems.  More spread out latent positions, due to disassorative dimensions, are likely to represent stronger connectivity. These results make sense, given that the cerebellum is responsible for balance and muscle control, which are commonly impaired in PD patients, and the occipital lobe is responsible for cognitive processing of visual information, which is also impaired in PD patients \citep{Weil2016VisualDisfPD, Gottlich2013PD}.  There are also visible differences in the temporal lobe (associated with memory and hearing), and the frontal lobe (motor control), which agrees well with previous reports in the literature \citep{LucasJimenez2016, Baggio2014}.  

We can also consider a change in relationships between different systems:  for example, the bigger separation between the cerebellum and frontal/temporal lobes in the PD group  suggests increased connectivity between these areas. This has been observed in the literature \citep{Tessitore2019} and is widely seen as a compensatory mechanism aimed at maintaining motor performance in the presence of basal ganglia dysfunction.



\begin{figure}[t!]
    \centering
    \includegraphics[width=0.98\linewidth]{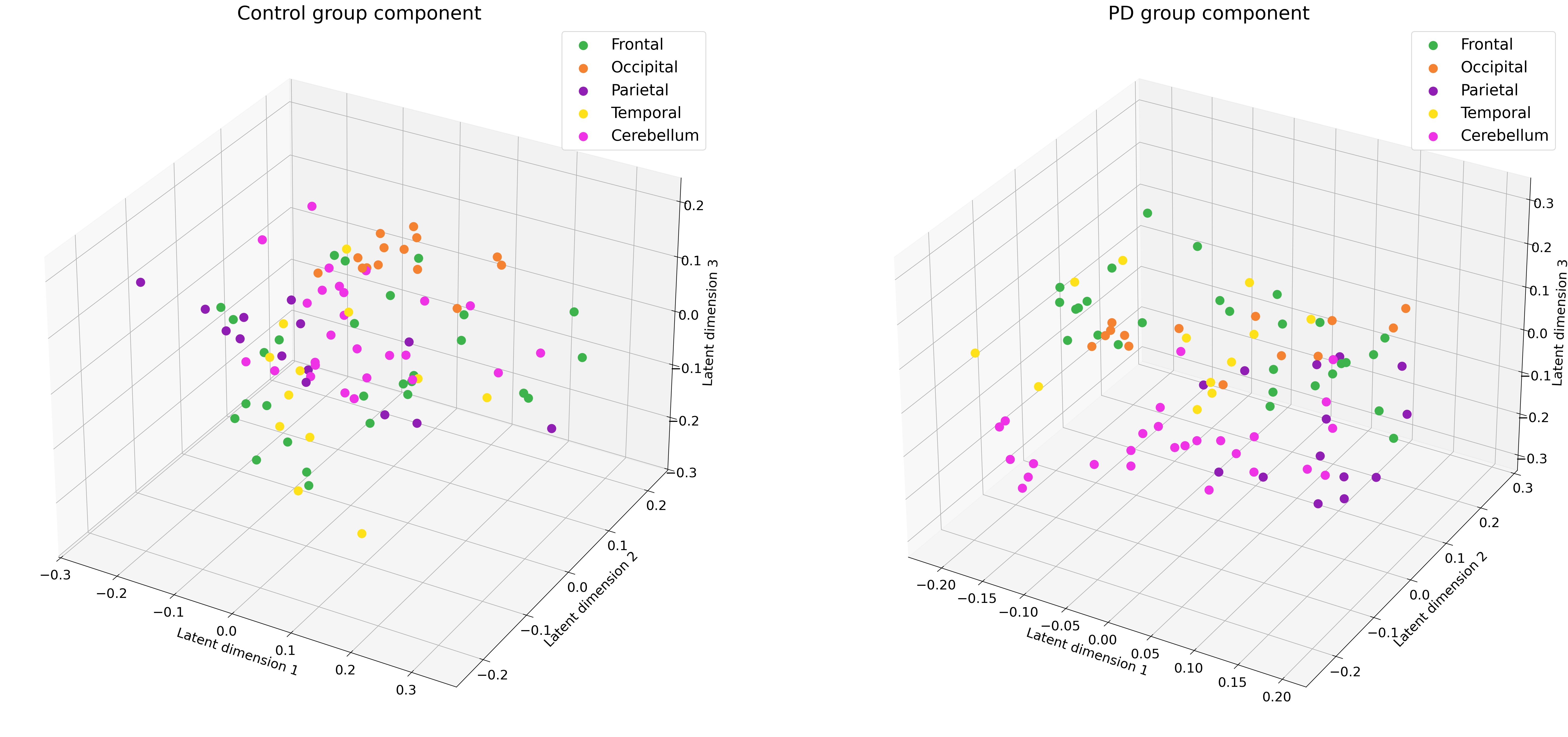}
    \caption{Three leading dimensions of the group latent positions estimated  with GroupMultiNeSS an(all disassortative).}
    \label{fig:gmn_group_components_taowu}
\end{figure}

\begin{figure}[t!]
    \centering
    \includegraphics[width=0.98\linewidth]{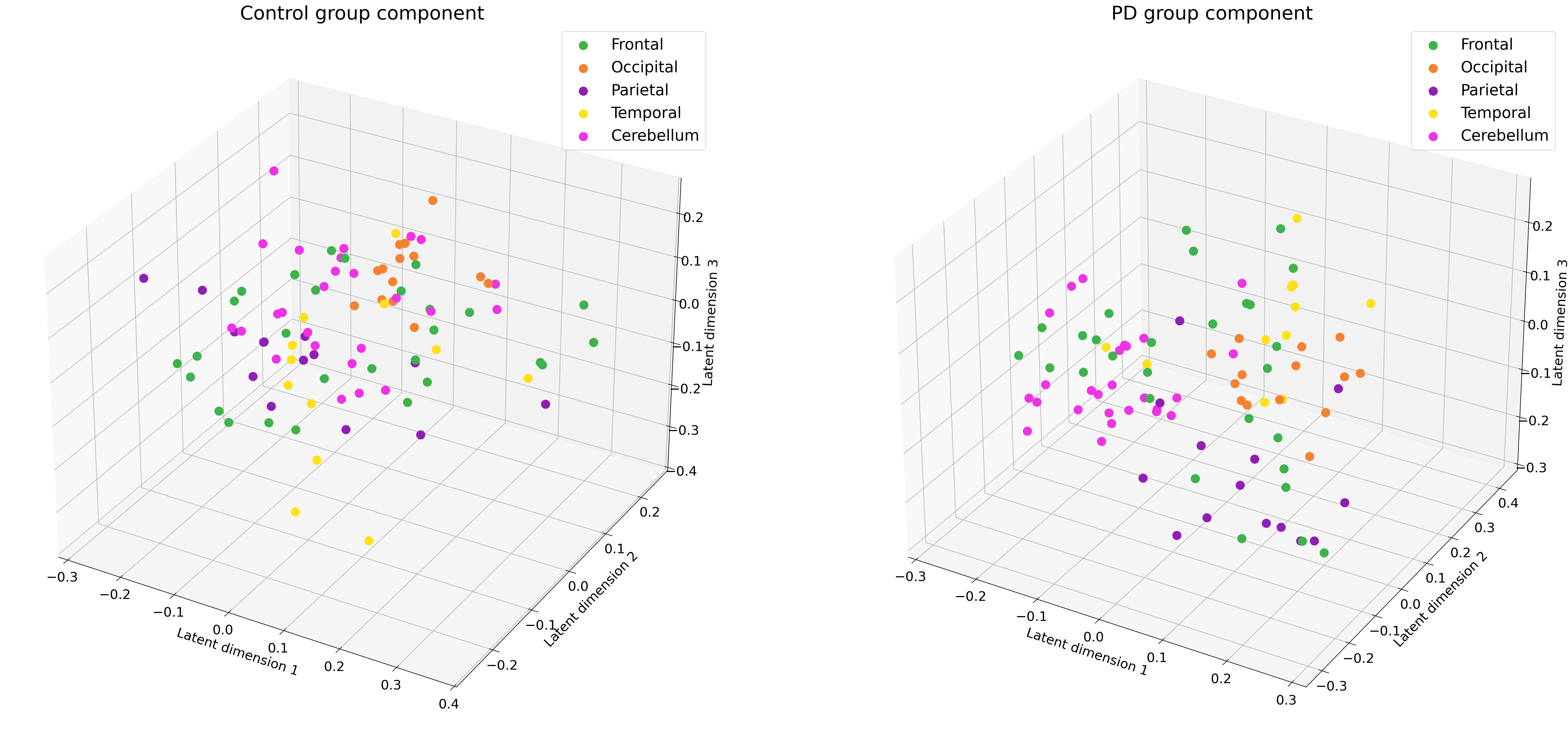}
    \caption{Group latent positions obtained by fitting separate MultiNeSS models on the layers of the two groups and plotted in the leading three latent dimensions (all disassortative).}
    \label{fig:mn_group_components_taowu}
\end{figure}

We also compare GroupMultiNeSS group components to extracting shared components from each group separately by running MultiNeSS on each.  
We expect that since the structure shared by all subjects in both groups is not separated out, the embeddings of the two groups will look more similar, and that is indeed what Figure \ref{fig:mn_group_components_taowu} shows. 

Finally, to quantify the differences between the two group components extracted by GroupMultiNeSS, we perform a permutation test.  For each pair of brain systems $(a, b)$, we consider all possible pairs $(r_1, r_2), r_1\in a, r_2 \in b$ and compute their unnormalized pairwise cosine similarity in the latent space of each group:
$$h_k^{(a, b)}(r_1, r_2) = \hat{W}_{k, r_1}^\top \hat{W}_{k, r_2}, \qquad k=1, 2.
$$ 
Averaging across all $r_1\in a, r_2 \in b$  gives $\bar{h}_k^{(a, b)}$, which can used as a proxy for connectivity between systems $a$ and $b$ in group $k$.  We can then look at the group differences $\bar{h}_2^{(a, b)} - \bar{h}_1^{(a, b)}$ for each pair of systems, shown as a heatmap in Figure \ref{fig:cos_sim_diff}.
Note that a positive difference corresponds to an increased average cosine similarity in the latent space of the PD group compared to the control, meaning lower lobe connectivity due to the disassortativeness of the latent dimensions. 

To informally assess the significance of these differences, we perform a permutation test by shuffling the group labels across groups 100 times to obtain the empirical distribution of the difference under the null hypothesis. 
We then apply the Benjamini-Hochberg correction to adjust for multiple testing. 

\begin{figure}[t!]
    \centering
    \includegraphics[width=0.7\linewidth]{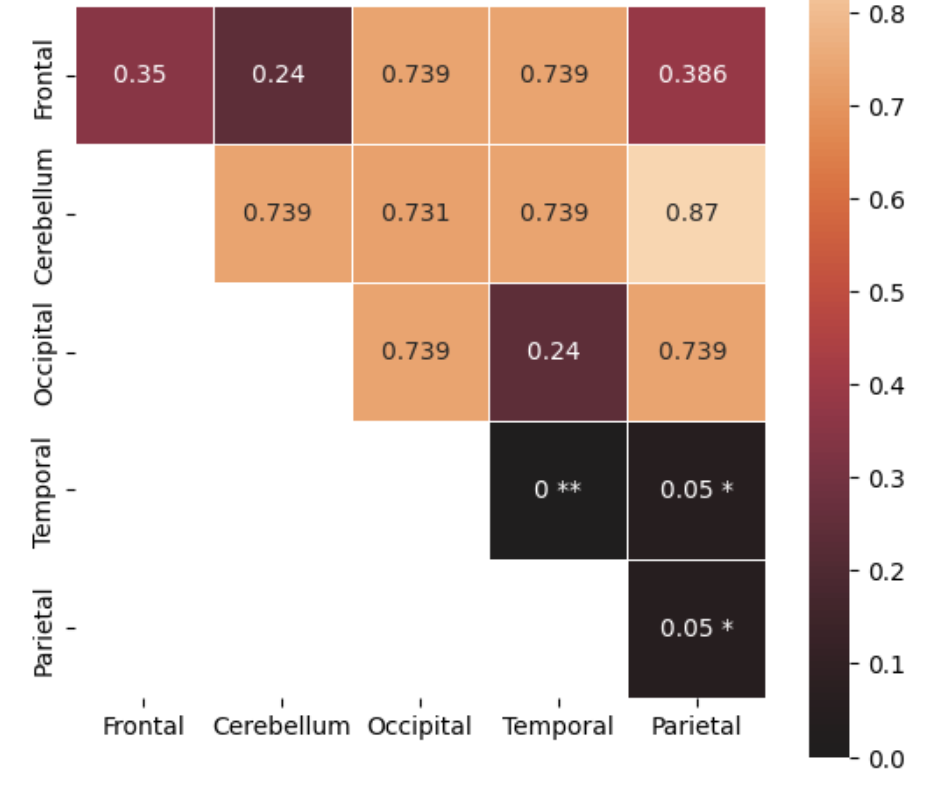}
    \caption{A heatmap of the BH-corrected p-values for the differences between the PD and control groups in the average cosine similarity of the brain systems (**:  $p < 0.01$; *:  $p < 0.05$).}
    \label{fig:cos_sim_diff}
\end{figure}

 Figure \ref{fig:cos_sim_diff} provides some quantitative support to the qualitative analysis of the visualizations in Figure \ref{fig:gmn_group_components_taowu}.  Fairly significant changes occur within and between the cerebellum, the occipital lobe, and the frontal lobe, which also stand out in the visualization.   Further, temporal and parietal lobes are implicated as the most significant, even after correcting for multiple testing, highlighting the need to develop more formal testing tools.   We leave this for future work.

\section{Discussion}
	\label{sec:discussion}
	The main contribution of this work is a latent space model that can explicitly separate out group-specific latent structure among multiplex network layers, distinct from both the common structure shared by all and from individual layer structure.   It is fully adaptive, in the sense that the algorithm learns from data how much the layers share through common or group-specific shared structure.  Both the estimation algorithm and theory can be directly extended to multiplex network models with a more complex group structure, for example, with a nested hierarchy of groups lay the basis for a universal modeling approach to complex multiplex networks. 


 We have already discussed the extension to the sub-Gaussian case for linear link functions in Remark \ref{gaus_relaxation_remark}.  Another natural extension of our work would be to apply the more general framework of \cite{TianYinqiu2024} to establish consistency for a broader class of link functions. Another important extension would be to consider the case of unknown group memberships, possibly fitting a mixture model or applying clustering to the layers as in \citep{pensky2024dimple}. An extension to directed networks could be developed by estimating the right and the left singular vectors separately. While we focused on estimation rather than testing in this work, the next step would be to develop more powerful tests for estimating group differences; perhaps adapting the framework of \cite{macdonald2024mesoscale} would allow for more formal comparisons between groups.

\bibliographystyle{chicago}
\bibliography{references}

@article{multiness,
	title = {Latent space models for multiplex networks with shared structure},
	volume = {109},
	issn = {1464-3510},
	url = {https://doi.org/10.1093/biomet/asab058},
	doi = {10.1093/biomet/asab058},
	pages = {683--706},
	number = {3},
	journal = {Biometrika},
	author = {MacDonald, P W and Levina, E and Zhu, J},
	year = {2022}
}

@article{mazumder,
  title={Spectral regularization algorithms for learning large incomplete matrices},
  author={Mazumder, Rahul and Hastie, Trevor and Tibshirani, Robert},
  journal={The Journal of Machine Learning Research},
  volume={11},
  pages={2287--2322},
  year={2010},
  publisher={JMLR. org}
}

@article{cai_zhang,
author = {Cai, T. and Zhang, Anru},
year = {2016},
pages = {},
title = {Rate-Optimal Perturbation Bounds for Singular Subspaces with Applications to High-Dimensional Statistics},
volume = {46},
journal = {The Annals of Statistics},
doi = {10.1214/17-AOS1541}
}

@article{Athreya2017,
author = {Athreya, Avanti and Fishkind, Donniell and Levin, Keith and Lyzinski, Vince and Park, Youngser and Qin, Yichen and Sussman, Daniel and Tang, Minh and Vogelstein, Joshua and Priebe, Carey},
year = {2017},
pages = {},
title = {Statistical inference on random dot product graphs: A survey},
volume = {18},
journal = {Journal of Machine Learning Research},
doi = {10.48550/arXiv.1709.05454}
}

@article{Young2007,
author="Young, Stephen J.
and Scheinerman, Edward R.",
editor="Bonato, Anthony
and Chung, Fan R. K.",
title="Random Dot Product Graph Models for Social Networks",
booktitle="Algorithms and Models for the Web-Graph",
year={2007},
journal = {Proceedings of the 5th International Conference on Algorithms and Models for the Web-Graph},
publisher="Springer Berlin Heidelberg",
address="Berlin, Heidelberg",
pages="138--149",
}

@article{Holland1983,
title = {Stochastic blockmodels: First steps},
journal = {Social Networks},
volume = {5},
number = {2},
pages = {109-137},
year = {1983},
issn = {0378-8733},
doi = {https://doi.org/10.1016/0378-8733(83)90021-7},
url = {https://www.sciencedirect.com/science/article/pii/0378873383900217},
author = {Paul W. Holland and Kathryn Blackmond Laskey and Samuel Leinhardt}
}

@article{li_network_2020,
	title = {Network cross-validation by edge sampling},
    year = {2020},
	volume = {107},
	issn = {0006-3444},
	url = {https://doi.org/10.1093/biomet/asaa006},
	doi = {10.1093/biomet/asaa006},
	pages = {257--276},
	number = {2},
	journal = {Biometrika},
	author = {Li, Tianxi and Levina, Elizaveta and Zhu, Ji},
	date = {2020-04},
}

@article{Arroyo2019,
author = {Arroyo, Jesús and Athreya, Avanti and Cape, Joshua and Chen, Guodong and Priebe, Carey and Vogelstein, Joshua},
year = {2021},
pages = {1-49},
title = {Inference for Multiple Heterogeneous Networks with a Common Invariant Subspace},
volume = {22},
journal = {Journal of machine learning research}
}

@article{gollini2016joint,
author = {Gollini, Isabella and Murphy, Thomas},
year = {2016},
pages = { 246--265},
title = {Joint Modeling of Multiple Network Views},
volume = {25},
number = 1,
journal = {Journal of Computational and Graphical Statistics},
doi = {10.1080/10618600.2014.978006}
}

@article{SalterTownshend2017,
 ISSN = {19326157},
 URL = {http://www.jstor.org/stable/26362225},
 author = {Michael Salter-Townshend and Tyler H. McCormick},
 journal = {The Annals of Applied Statistics},
 number = {3},
 pages = {1217--1244},
 publisher = {Institute of Mathematical Statistics},
 title = {LATENT SPACE MODELS FOR MULTIVIEW NETWORK DATA},
 urldate = {2025-01-01},
 volume = {11},
 year = {2017}
}

@article{DAnegelo2019,
author = {D'Angelo, Silvia and Murphy, Thomas and Alfo, Marco},
year = {2018},
pages = {},
title = {Latent Space Modeling of Multidimensional Networks with Application to the Exchange of Votes in Eurovision Song Contest},
volume = {13},
journal = {The Annals of Applied Statistics},
doi = {10.1214/18-AOAS1221}
}

@misc{Sosa2021,
      title={A Latent Space Model for Multilayer Network Data}, 
      author={Juan Sosa and Brenda Betancourt},
      year={2021},
      note={arXiv:2102.09560},
      primaryClass={cs.SI},
      url={https://arxiv.org/abs/2102.09560}, 
}

@article{FithianMazumder2018,
author = {Fithian, William and Mazumder, Rahul},
year = {2018},
equation={36},
pages = {238-260},
title = {Flexible Low-Rank Statistical Modeling with Missing Data and Side Information},
volume = {33},
journal = {Statistical Science},
doi = {10.1214/18-STS642}
}

@article{Koltchinskii2010,
author = {Koltchinskii, Vladimir and Tsybakov, Alexandre and Lounici, Karim},
year = {2010},
pages = {},
title = {Nuclear norm penalization and optimal rates for noisy low rank matrix
completion},
volume = {39},
journal = {Annals of Statistics},
doi = {10.1214/11-AOS894}
}

@article{TianYinqiu2024,
      title={Efficient Analysis of Latent Spaces in Heterogeneous Networks}, 
      author={Yuang Tian and Jiajin Sun and Yinqiu He},
      year={2024},
      note={arXiv:2412.02151},
      primaryClass={stat.ME},
      url={https://arxiv.org/abs/2412.02151}, 
}

@article{TaoWu2017,
    doi = {10.1371/journal.pone.0188196},
    author = {Badea, Liviu AND Onu, Mihaela AND Wu, Tao AND Roceanu, Adina AND Bajenaru, Ovidiu},
    journal = {PLOS ONE},
    publisher = {Public Library of Science},
    title = {Exploring the reproducibility of functional connectivity alterations in {P}arkinson’s disease},
    year = {2017},
    volume = {12},
    url = {https://doi.org/10.1371/journal.pone.0188196},
    pages = {1-21},
    number = {11},
}

@article{Hoff2002,
author = {Peter D Hoff and Adrian E Raftery and Mark S Handcock},
title = {Latent Space Approaches to Social Network Analysis},
journal = {Journal of the American Statistical Association},
volume = {97},
number = {460},
pages = {1090--1098},
year = {2002},
publisher = {ASA Website},
doi = {10.1198/016214502388618906},
URL = {https://doi.org/10.1198/016214502388618906},
eprint = {https://doi.org/10.1198/016214502388618906
}

}

@article{Jones2021MRDPG,
      title={The multilayer random dot product graph}, 
      author={Andrew Jones and Patrick Rubin-Delanchy},
      year={2021},
      note={arXiv:2007.10455},
      primaryClass={stat.ML},
      url={https://arxiv.org/abs/2007.10455}, 
}

@article{Bandeira2016,
   title={Sharp nonasymptotic bounds on the norm of random matrices with independent entries},
   volume={44},
   ISSN={0091-1798},
   url={http://dx.doi.org/10.1214/15-AOP1025},
   DOI={10.1214/15-aop1025},
   number={4},
   journal={The Annals of Probability},
   publisher={Institute of Mathematical Statistics},
   author={Bandeira, Afonso S. and van Handel, Ramon},
   year={2016}}

@article{AAL116_atlas,
title = {Automated Anatomical Labeling of Activations in SPM Using a Macroscopic Anatomical Parcellation of the MNI MRI Single-Subject Brain},
journal = {NeuroImage},
volume = {15},
number = {1},
pages = {273-289},
year = {2002},
issn = {1053-8119},
doi = {https://doi.org/10.1006/nimg.2001.0978},
url = {https://www.sciencedirect.com/science/article/pii/S1053811901909784},
author = {N. Tzourio-Mazoyer and B. Landeau and D. Papathanassiou and F. Crivello and O. Etard and N. Delcroix and B. Mazoyer and M. Joliot}
}

@article{Weil2016VisualDisfPD,
  title     = "Visual dysfunction in {P}arkinson's disease",
  author    = "Weil, Rimona S and Schrag, Anette E and Warren, Jason D and
               Crutch, Sebastian J and Lees, Andrew J and Morris, Huw R",
  journal   = "Brain",
  publisher = "Oxford University Press (OUP)",
  volume    =  139,
  number    =  11,
  pages     = "2827--2843",
  year      =  2016,
  language  = "en"
}

@article{Gottlich2013PD,
  title     = "Altered resting state brain networks in {P}arkinson's disease",
  author    = "G{\"o}ttlich, Martin and M{\"u}nte, Thomas F and Heldmann,
               Marcus and Kasten, Meike and Hagenah, Johann and Kr{\"a}mer,
               Ulrike M",
  journal   = "PLoS One",
  publisher = "Public Library of Science (PLoS)",
  volume    =  8,
  number    =  10,
  pages     = "e77336",
  year      =  2013,
  language  = "en"
}

@article{DiffADHD2017,
  title    = "Functional brain connectivity differences between different {ADHD} presentations: Impaired functional segregation in {ADHD-combined} presentation but not in {ADHD-inattentive} presentation",
  author   = "Ghaderi, Amir Hossein and Nazari, Mohammad Ali and Shahrokhi,
              Hassan and Darooneh, Amir Hossein",
  journal  = "Basic Clin. Neurosci.",
  volume   =  8,
  number   =  4,
  pages    = "267--278",
  year     =  2017,
  language = "en"
}

@ARTICLE{DiffAlzheimer2024,
  title     = "Changes of brain functional network in Alzheimer's disease and
               frontotemporal dementia: a graph-theoretic analysis",
  author    = "Wu, Shijing and Zhan, Ping and Wang, Guojing and Yu, Xiaohua and
               Liu, Hongyun and Wang, Weidong",
  journal   = "BMC Neurosci.",
  publisher = "Springer Science and Business Media LLC",
  volume    =  25,
  number    =  1,
  pages     = "30",
  year      =  2024,
  copyright = "https://creativecommons.org/licenses/by/4.0",
  language  = "en"
}

@inproceedings{han2015consistent,
  author    = {Han, Qiuyi and Xu, Kevin S. and Airoldi, Edoardo M.},
  title     = {Consistent estimation of dynamic and multi-layer block models},
  booktitle = {Proceedings of the 32nd International Conference on Machine Learning},
  year      = {2015}
}

@article{pensky2024dimple,
author = {Pensky, Marianna and Wang, Yaxuan},
year = {2024},
month = {07},
pages = {1-14},
title = {Clustering of Diverse Multiplex Networks},
volume = {PP},
journal = {IEEE Transactions on Network Science and Engineering},
doi = {10.1109/TNSE.2024.3374102}
}

@article{nguen2024networktwosampletestblock,
      title={Network two-sample test for block models}, 
      author={Chung Kyong Nguen and Oscar Hernan Madrid Padilla and Arash A. Amini},
      year={2024},
      note={arXiv:2406.06014},
      primaryClass={math.ST},
      url={https://arxiv.org/abs/2406.06014}, 
}

@inproceedings{xuefei2020flexible,
author = {Zhang, Xuefei and Xue, Songkai and Zhu, Ji},
title = {A flexible latent space model for multilayer networks},
year = {2020},
booktitle = {Proceedings of the 37th International Conference on Machine Learning},
articleno = {1047},
numpages = {10},
}

@article{TangSemiparamNetworkTest,
    author = {Minh Tang and Avanti Athreya and Daniel L. Sussman and Vince Lyzinski and Youngser Park and Carey E. Priebe},
    title = {A Semiparametric Two-Sample Hypothesis Testing Problem for Random Graphs},
    journal = {Journal of Computational and Graphical Statistics},
    volume = {26},
    number = {2},
    pages = {344--354},
    year = {2017},
    publisher = {ASA Website},
    doi = {10.1080/10618600.2016.1193505},
    
    
    URL = { https://doi.org/10.1080/10618600.2016.1193505
    },
    eprint = {https://doi.org/10.1080/10618600.2016.1193505
    }
}

@article{macdonald2024mesoscale,
      title={Mesoscale two-sample testing for network data}, 
      author={Peter W. MacDonald and Elizaveta Levina and Ji Zhu},
      year={2024},
      note={arXiv:2410.17046},
      primaryClass={stat.ME},
      url={https://arxiv.org/abs/2410.17046}, 
}

@ARTICLE{LucasJimenez2016,
  title    = "Altered functional connectivity in the default mode network is
              associated with cognitive impairment and brain anatomical changes
              in {P}arkinson's disease",
  author   = "Lucas-Jim{\'e}nez, Olaia and Ojeda, Natalia and Pe{\~n}a, Javier and D{\'\i}ez-Cirarda, Mar{\'\i}a and Cabrera-Zubizarreta,
              Alberto and G{\'o}mez-Esteban, Juan Carlos and
              G{\'o}mez-Beldarrain, Mar{\'\i}a {\'A}ngeles and
              Ibarretxe-Bilbao, Naroa",
  journal  = "Parkinsonism Relat. Disord.",
  volume   =  33,
  pages    = "58--64",
  year     =  2016,
  language = "en"
}

@ARTICLE{Tessitore2019,
  title     = "Functional connectivity signatures of {P}arkinson's disease",
  author    = "Tessitore, Alessandro and Cirillo, Mario and De Micco, Rosa",
  journal   = "J. Parkinsons. Dis.",
  publisher = "IOS Press",
  volume    =  9,
  number    =  4,
  pages     = "637--652",
  year      =  2019,
  language  = "en"
}

@ARTICLE{Baggio2014,
  title     = "Functional brain networks and cognitive deficits in {P}arkinson's
               disease",
  author    = "Baggio, Hugo-Cesar and Sala-Llonch, Roser and Segura,
               B{\`a}rbara and Marti, Maria-Jos{\'e} and Valldeoriola, Francesc
               and Compta, Yaroslau and Tolosa, Eduardo and Junqu{\'e}, Carme",
  journal   = "Hum. Brain Mapp.",
  publisher = "Wiley",
  volume    =  35,
  number    =  9,
  pages     = "4620--4634",
  year      =  2014,
  copyright = "http://onlinelibrary.wiley.com/termsAndConditions\#vor",
  language  = "en"
}

@article{MaMaUnifModel2020,
  author  = {Zhuang Ma and Zongming Ma and Hongsong Yuan},
  title   = {Universal Latent Space Model Fitting for Large Networks with Edge Covariates},
  journal = {Journal of Machine Learning Research},
  year    = {2020},
  volume  = {21},
  number  = {4},
  pages   = {1--67},
  url     = {http://jmlr.org/papers/v21/17-470.html}
}

\appendix
	\label{ch:appendix}
	

\section{Refitting step details}\label{refitting_details_section}

The drawback of using the nuclear norm penalty is the bias caused by shrinking the non-zero eigenvalues of the fitted matrices;  however, this bias does not affect the estimated eigenvectors.  The standard remedy for this is to use the refitting procedure developed by \cite{mazumder}.

We begin by describing the $\textit{FirstStageRefit}$, that is, the refitting step for Problem \eqref{first-stage-conv-prob}. Consider the eigen-decompositions of the solutions to Problem \eqref{first-stage-conv-prob}:
$$
\widehat{S + Q_k}=\hat{\bar{V}}_{0k} \hat{\Gamma}_{0k} \hat{\bar{V}}_{0k}^{\top},\quad \hat{R}_{k \ell}=\hat{\bar{U}}_k \hat{\Gamma}_{k \ell} \hat{\bar{U}}_{k \ell}^{\top}, \quad \text{where}\quad  \ell=1,\ldots, m_k.
$$
 Element-wise, we have
$$
[\widehat{S +Q_k}]_{i j}=\sum_{r=1}^{\widehat{d_0 + d_k}} \gamma_{r}(\widehat{S + Q_k}) \hat{\bar{V}}_{0k, i r} \hat{\bar{V}}_{0k,j r} \quad i=1, \ldots, n ;\ j=1, \ldots n,
$$
 where $\widehat{d_0 + d_k}$ is the estimated rank of $\widehat{S + Q_k}$, and $\gamma_{r}(\cdot)$ denotes the $r$-th eigenvalue of the input matrix, ordered by magnitude.  We can write the element-wise decomposition of $\hat{R}_{k\ell}$ similarly.
 Since $\widehat{S + Q_k}$ and $\{\hat{R}_{k \ell}\}_{\ell=1}^{m_k}$ are expected to be low-rank up to numerical precision (due to the nuclear norm penalization), in our implementation, we compute the ranks by counting the number of eigenvalues with magnitude above a small positive threshold, set to $10^{-6}$. Fixing the eigenvectors corresponding to such eigenvalues, $\textit{FirstStageRefit}$ solves the following convex problem with variables $\hat{\Gamma}_{0k}$ and $\{\hat{\Gamma}_{k \ell}\}_{\ell=1}^{m_k}$:
\begin{equation}\label{first_stage_refit}
\begin{aligned}
\min\ \Bigl\{-\sum_{\ell=1}^{m_k} \sum_{i\le j} \log f\Bigl(A_{k \ell, i j} ; & \sum_{r=1}^{\widehat{d_0 + d_k}} \gamma_{r}(\widehat{S + Q_k}) \hat{\bar{V}}_{0k,ir} \hat{\bar{V}}_{0k,jr} \ + \\ 
& \sum_{r=1}^{\hat{d}_{k \ell}} \gamma_{r}(\hat{R}_k) \hat{\bar{U}}_{k \ell, i r} \hat{\bar{U}}_{k \ell, j r}, \ \phi\Bigr)\Bigr\}
\end{aligned}
\end{equation}
 
Similarly, consider the eigen-decompositions of the solutions to Problem \eqref{second-stage-conv-prob} that are truncated up to the first $\hat{d}_0$ and $\hat{d}_k, k=1,\ldots, K$ eigenvectors, respectively:
$$
\hat{S}=\hat{\bar{V}} \hat{\Gamma}_0 \hat{\bar{V}}^{\top}, \quad \hat{Q}_k =\hat{\bar{W}}_k \hat{\Gamma}_k\hat{\bar{W}}_k ^{\top}, \quad \text{where}\quad  k=1,\ldots, K.
$$
%
Then $\textit{SecondStageRefit}$ solves the following convex problem with variables $\hat{\Gamma}_0, \{\hat{\Gamma}_k\}_{k=1}^K$ and the refitted estimates $\{\hat{R}_{k\ell}\}_\mathcal{I}$ of the individual components kept fixed:

\begin{equation}\label{second_stage_refit}
\begin{aligned}
\min \ \Bigl\{-\sum_{k=1}^K\sum_{\ell=1}^{m_k} \sum_{i\le j} \log f\Bigl(A_{k \ell, i j} ; & \sum_{r=1}^{\hat{d}_0} \gamma_{r}(\hat{S}) \hat{\bar{V}}_{ir} \hat{\bar{V}}_{jr} \ + \\ 
& \sum_{r=1}^{\hat{d}_{k}} \gamma_{r}(\hat{Q}_k) \hat{\bar{W}}_{k, i r} \hat{\bar{W}}_{k, j r} \ + \hat{R}_{k\ell, ij}, \ \phi\Bigr)\Bigr\}.
\end{aligned}
\end{equation}
 
If $f$ is from a one-parameter exponential family as in \eqref{one_param_exp_family}, $\textit{FirstStageRefit}$ is  equivalent to fitting a GLM with $\widehat{d_0 + d_k}  + \sum_{\ell =1}^{m_k} \hat{d}_{k \ell}$ predictors (no intercept) and $n(n+1) m_k / 2$ responses. In turn, Problem \eqref{second_stage_refit} is equivalent to fitting a GLM with $\hat{d_0}  + \sum_{k=1}^{K} \hat{d}_{k}$ predictors and $n(n+1) M / 2$ responses. Notice that the first-stage GLMs are fitted without any intercept, while the second-stage GLM is fitted with a fixed vector of observation-dependent offsets $\hat{R}_{k\ell,ij}$. 


\section{Proofs}
\subsection{Section \ref{sec:identifiability} proofs (Identifiability)}\label{identif_proof_section} 
\begin{proof}[Proof of Proposition \ref{identif_propos}]
    By condition 1 and Lemma 1 in \citep{multiness}, for any $k=1, \ldots, K$ and $\ell = 1, \ldots, m_k$, we have
    \begin{equation}\label{identif_rotation}
        [V \ W_k\ U_{k\ell}] \ O_{k\ell} = [V'\ W_k'\ U_{k\ell}']
    \end{equation}
    where $O_{k\ell}$ satisfies $S_{k\ell}O_{k\ell}S_{k\ell}^\top \in \mathcal{O}_{p_0 + p_k + p_{k\ell}, \ q_0 + q_k + q_{k\ell}}$ with a permutation matrix
    $$S_{k\ell}=\left[\begin{array}{cccccc}
    I_{p_0} & 0 & 0 & 0 & 0 & 0\\
    0 & 0 & 0 & I_{q_0} & 0 & 0\\
    0 & I_{p_k} & 0 & 0 & 0 & 0 \\
    0 & 0 & 0 & 0 & I_{q_k} & 0\\
    0 & 0 & I_{p_{k\ell}} & 0 & 0 & 0 \\
    0 & 0 & 0 & 0 & 0 & I_{q_{k\ell}} 
    \end{array}\right].
    $$
    It is sufficient to demonstrate that $O_{k\ell}$ is block-diagonal for any layer, i.e., we need to verify that
    $$O_{k\ell} = \left[\begin{array}{ccc}
    O_{11, k\ell} & O_{12, k\ell} &    O_{13, k\ell} \\ 
    O_{21, k\ell} & O_{22, k\ell} &    O_{23, k\ell} \\
    O_{31, k\ell} & O_{32, k\ell} &    O_{33, k\ell}  \end{array}\right] = 
    \left[\begin{array}{ccc}
    O_{11, k\ell} & 0 &    0 \\ 
   0 & O_{22, k\ell} &    0 \\
   0 & 0 &    O_{33, k\ell} 
    \end{array}\right]
    $$
    with $O_{11, k\ell} \in \mathcal{O}_{p_0, q_0}, O_{22, k\ell} \in \mathcal{O}_{p_k, q_k}, O_{33, k\ell} \in \mathcal{O}_{p_{k\ell}, q_{k\ell}}$.
    We will refer to the components of the $r$-th column of $O_{k\ell}$ as follows $$o_{k\ell, r} = \left[\begin{array}{l}
o_{k\ell, r}^{(1)} \\
o_{k\ell, r}^{(2)} \\
o_{k\ell, r}^{(3)} 
\end{array}\right]\in \mathbb{R}^{d_0} \times \mathbb{R}^{d_k} \times \mathbb{R}^{d_{k\ell}}
$$

Consider layers $(k_1, \ell_1)$ and $(k_2, \ell_2)$ satisfying condition 3. 
Plugging them into \eqref{identif_rotation} and equating two expressions for $V'$, we obtain, for every $ r = 1, \ldots, d_0$
$$0 = V(o_{k_1\ell_1, r}^{(1)} - o_{k_2\ell_2, r}^{(1)}) + W_{k_1}o_{k_1\ell_1, r}^{(2)} - W_{k_2}o_{k_2\ell_2, r}^{(2)} + U_{k_1\ell_1}o_{k_1\ell_1, r}^{(3)} - U_{k_2\ell_2}o_{k_2\ell_2, r}^{(3)}.
$$
By linear independence, it implies that $o_{k_1\ell_1, r}^{(2)} = o_{k_1\ell_1, r}^{(3)} = 0$. Since it holds for any $r=1, \ldots, d_0$, we have
$O_{21, k_1\ell_1} = O_{31, k_1\ell_1} = 0$. 
By \eqref{identif_rotation}, it means $V' = VO_{k_1, \ell_1}$ and so for any layer $(k, \ell)$, we have $O_{21, k\ell} = O_{31, k\ell} = 0$.
Using the fact that $S_{k\ell}O_{k\ell}S_{k\ell}^\top$ is an indefinite orthogonal transformation, we can conclude that the symmetric upper-diagonal blocks $O_{12, k\ell}, O_{13, k\ell}$ are also zeros.

Now, consider layers $s$ and $t$ in group $k$ satisfying the second identifiability condition. Plugging them into \eqref{identif_rotation} and equating two expressions for $W_k'$ now, we have, for every $r=d_0 + 1, \ldots, d_0 + d_k$
$$0 = V (o_{ks, r}^{(1)} - o_{kt, r}^{(1)}) + W_k (o_{ks, r}^{(2)} - o_{kt, r}^{(2)}) + U_{ks}o_{ks, r}^{(3)} - U_{ks}o_{kt, r}^{(3)}.
$$
By linear independence, $o_{ks, r}^{(3)} = o_{kt, r}^{(3)} = 0$. Since this holds for all $r=d_0 + 1, \ldots, d_0 + d_k$, we deduce
$O_{32, ks} = 0$. Combined with $O_{12, ks} = 0$, this implies that $W_k' = W_kO_{22, ks}$ and so for any layer $(k, \ell)$, we have $O_{32, k\ell} = 0$. Since the lower-right $2\times 2$ block sub-matrix of $O_{k\ell}$ is also an indefinite orthogonal rotation, we deduce $O_{23, k\ell} = 0$, which completes the proof.
\end{proof}

\subsection{Section \ref{sec:estimation} Proofs (Optimization)}\label{fit_section_proofs}

\begin{proof}[Proof of Proposition \ref{second_stage_gaus_case_proposition}]
    For the Gaussian distribution, for each $1\le i\le j \le n$, we can rewrite the joint negative log-likelihood of the $(i, j)$-th entries (up to the $1/2\sigma^2$ multiplier) as
    \begin{equation*}
    \begin{aligned}
        \sum_{(k, \ell)\in\mathcal{I}} (A_{k\ell, ij} - S_{ij} -  Q_{k, ij} - \hat{R}_{k\ell, ij})^2 
        &=\sum_{k=1}^K m_k \bigl(S_{ij} + Q_{k, ij} -\tilde{A}_{k,ij}\bigr)^2 + const,
    \end{aligned}
    \end{equation*}
    where $const$ is a term independent of $S$ and $Q_k$. Therefore, the solution to Problem \eqref{second-stage-conv-prob} coincides with the solution of  
    $$\min_{S, Q_k}\ \Bigl\{\sum_{k=1}^{K}{m_k\over 2\sigma^2}\sum_{i\le j}\bigl[S_{ij} + Q_{k, ij} - \tilde{A}_{k, ij}\bigr]^2 + \lambda_2 \|S\|_* + \sum_{k=1}^K\lambda_2\alpha_{2k} \|Q_k\|_* \Bigr\}.
    $$
\end{proof}

\begin{proof}[Proof of Proposition \ref{SpQ_avg_residual_relationship_propos}]
    With the assumed edge distribution, Problem \eqref{first-stage-conv-prob} solves \begin{equation}\label{loop_free_gaus_frist_stage}
    \widehat{S + Q_k}, \{\hat{R}_{k\ell}\}_{\ell = 1}^{m_k} = \argmin_{S + Q_k, R_{k\ell}} {1\over4\sigma^2}\sum_{\ell = 1}^{m_k}\Bigl\|A_{k\ell} - (S + Q_k) - R_{k\ell}\Bigr\|_F^2 +\ \lambda_{1k}\| S + Q_k\|_* + \sum_{\ell=1}^{m_k}\lambda_{1k} \alpha_{1k\ell}\|R_{k\ell}\|_*,
\end{equation}
    where the off-diagonal entries of each matrix inside the Frobenius norm are counted twice compared to \eqref{group_log_likelihood}, leading to an additional $1/2$ multiplier in front of the sum. By optimality, $\widehat{S + Q_k}$ should minimize the target function in \eqref{loop_free_gaus_frist_stage} over $S + Q_k$ with each $R_{k\ell}$ fixed to $\hat{R}_{k\ell}$, that is, 
    \begin{align*}
        \widehat{S+Q}_k =& \argmin_{Z\in\mathbb{R}^{n\times n}} \Bigl\{{1\over 4\sigma^2}\sum_{\ell=1}^{m_k}\bigl\|\bigl(A_{k\ell} - \hat{R}_{k\ell}\bigr) - Z\bigr\|_F^2 + \lambda_{1k} \bigl\|Z\bigr\|_* \Bigr\}\\
        =&\argmin_{Z\in\mathbb{R}^{n\times n}} \Bigl\{{1\over 2}\bigl\|{1\over m_k}\sum_{\ell=1}^{m_k}\bigl(A_{k\ell} - \hat{R}_{k\ell}\bigr) - Z\bigr\|_F^2 + {2\sigma^2\lambda_{1k} \over m_k}\bigl\|Z\bigr\|_* + const \Bigr\}
    \end{align*}
where $const$ is a term independent of $Z$. By the variational definition of soft-thresholding, the optimal $Z$ in the last problem is achieved at \eqref{SpQ_avg_residual_relationship}.
\end{proof}

\subsection{Section \ref{sec:theory} Proofs (Consistency)}\label{consist_section_proofs}
We start with additional notation. For $a, b\in\mathbb{R}$, let $a\wedge b = \min(a, b), \ a\vee b = \max(a, b)$. For a matrix $V\in\mathbb{R}^{n\times d}$, denote the projection operator onto its column space by $\mathcal{P}_V = V(V^\top V)^\dagger V^\top$,  where $(\cdot)^\dagger$ is the Moore-Penrose pseudoinverse. For a symmetric matrix $Z\in\mathbb{R}^{n\times n}$ with the eigendecomposition $Z = V\Gamma V^\top$, we always order the diagonal entries of  $\Gamma$ in descending order by their absolute values, with  
$$\|\Gamma\|_2 = |\gamma_1(Z)| \ge \ldots \ge |\gamma_n(Z)| = \gamma_{\min}(Z) . 
$$

In this section, we use multiple measures of similarity between subspaces spanned by the columns of two $n \times d$ orthogonal matrices $V$ and $\hat{V}$. Suppose the singular values of $V^{\top} \hat{V}$ are $\sigma_1 \geq$ $\sigma_2 \geq \cdots \geq \sigma_d \geq 0$. Then, following standard notations, we define the $\sin \Theta$-matrix as
\begin{equation*}
\sin\Theta(V, \hat{V})=\operatorname{diag}\left[\sin(\cos ^{-1}\left(\sigma_1\right)), \cdots, \sin(\cos ^{-1}\left(\sigma_d\right))\right], 
\end{equation*} 
and use $\|\sin \Theta(V, \hat{V})\|_2$ to measure the similarity between $V$ and $\hat V$. Further, for symmetric matrices $Z,\hat{Z}\in\mathbb{R}^{n\times n}$ with leading $d < n$ eigenvectors denoted  as $V, \hat{V} \in\mathbb{R}^{n\times d}$, respectively, we define $\sin_d(Z, \hat{Z}):= \|\sin \Theta(V, \hat{V})\|_2$.   This is uniquely defined only if either $|\gamma_d| > |\gamma_{d+1}|$ or $|\gamma_d| = |\gamma_{d+1}| = 0$.  For all matrices we use this for, this will be true for sufficiently large $n$. 

Lemma 1 in \cite{cai_zhang} establishes the equivalence between the sin distance and several other metrics. We will use the following two corollaries from this Lemma:    
\begin{equation}\label{vvT_sin_theta}
\|\hat{V} \hat{V}^{\top}-V V^{\top}\|_2 \leq 2\|\sin \Theta(\hat{V}, V)\|_2,
\end{equation}
\begin{equation}\label{infO_sin_theta}
\inf_{O\in \mathcal{O}_d}\|\hat{V} - VO\|_2 \leq \sqrt{2}\|\sin \Theta(\hat{V}, V)\|_2.
\end{equation}
Another important result we use is Theorem 1 of \cite{cai_zhang}, which provides a useful bound on the $\sin$-distance between the eigenspaces of a matrix and its perturbation.  Since originally this result was formulated for arbitrary matrices and we only need it for symmetric matrices, we restate it in Lemma \ref{version_cai_zhang_thm} with a slightly looser but much more convenient bound.



\begin{proof}[Proof of Theorem \ref{main_consistency_theorem}]
Assume $\mathcal{E}_{noise}$ in \eqref{error_bound_set} and initialization error event in \eqref{init_error_set} occur simultaneously, which has probability at least $1 - (M + K + 1) n e^{-C_0n}$ by Lemma \ref{avg_errors_lemma} and Assumption \ref{initializer_assump}.  

The initial step is to rewrite the first iteration update for each parameter in Stage I and Stage II as a soft-thresholding operator applied to the parameter's ground-truth value, say $Z$, plus the error term, which comprises Gaussian noise and the weighted sum of an update-dependent set of latent components errors:
\begin{equation}\label{intuitive_soft_threshold_update}
    \hat{Z} = \mathcal{T}_\rho(Z + \text{other components' errors} + \text{Gaussian noise}).
\end{equation}
The first stage update rule for the parameters $S+Q_k$ and $\{R_{k\ell}\}_{\ell=1}^{m_k}$ can be rewritten 
\begin{equation}\label{first_stage_gaus_with_error}
    \begin{aligned}
    &R_{k\ell}^{(1)}
     = \mathcal{T}_{\rho_{1k\ell}} \Bigl[R_{k\ell}  - \Delta_{S+Q_k}^{(0)} + E_{k\ell}\Bigr], \qquad \text{for } \ell=1,\ldots, m_k, \\
      &(S + Q_k)^{(1)}  = \mathcal{T}_{\rho_{1k}} \Bigl[(S + Q_k) - {1\over m_k} \sum_{\ell=1}^{m_k} \Delta_{R_{k\ell}}^{(1)} + \bar{E}_k\Bigr], 
    \end{aligned}
\end{equation}
where by $\Delta_Z^{(1)}:=Z^{(1)} - Z$ we denote the error of parameter $Z$ after the first iteration.
The second stage update rule at the first iteration can be written  similarly:
\begin{equation}\label{second_stage_gaus_with_error}
    \begin{aligned}
    &Q_{k}^{(1)} 
     = \mathcal{T}_{\rho_{2k}}
     \Bigl[Q_k - \Delta_{S}^{(0)} - {1\over m_k} \sum_{\ell=1}^{m_k} \Delta_{R_{k\ell}} + \bar{E}_k \Bigr], \qquad k=1, \ldots, K,\\
      &S^{(1)} =\mathcal{T}_{\rho_2} \Bigl[S - \sum_{k=1}^K {m_k\over M}\Delta_{Q_k}^{(1)}  -  {1\over M} \sum_{(k,\ell)\in\mathcal{I}}\Delta_{R_{k\ell}} + \bar{E}\Bigr], 
    \end{aligned}
\end{equation}
where $\Delta_{R_{k\ell}}, (k, \ell) \in \mathcal{I}$ denotes the individual component error after the $k$-th group updates in the first stage.  

The key to our proof will be Lemma \ref{soft_threshold_properties} that shows that the soft threshold operator applied to a matrix perturbed by additive error as in \eqref{intuitive_soft_threshold_update} should have the threshold $\rho$ with the rate of the total error's spectral norm to guarantee that the estimate $\hat{Z}$ is close to the ground truth $Z$ in the Frobenius norm. In particular, it should dominate both the spectral norm of the other components' errors, which we control using Lemma \ref{avg_errors_lemma},  and the average of the appropriate Gaussian noise matrices, which we control by restricting the analysis to the set $\mathcal{E}_{noise}$. In what follows, we formalize this intuition for Stage I and then for Stage II.\\

\noindent {\bf Stage I (group $k$):} 
We first establish the properties of $R_{k\ell}^{(1)}, \ \ell=1, \ldots, m_k$ by applying Lemma \ref{soft_threshold_properties} with 
$$E :=  E_{k\ell} -\Delta_{S+ Q_k}^{(0)}, \quad  \rho: =\rho_{1k\ell}, \quad d := d_{k\ell}.$$
Define $\rho_{1k\ell}=C_1( r^{(I)}_k\vee n^{1/2})$ with $C_1:=2(1 + 3\sigma)$,
By Assumption \ref{initializer_assump} and the triangular inequality, we have  
\begin{equation}\label{total_error_bound_first_stage}
    \|E_{k\ell} - \Delta_{S+ Q_k}^{(0)}\|_2 \le r^{(I)}_k + 3\sigma n^{1/2} \le \rho_{1k\ell} / 2 . 
\end{equation}
Therefore,  Properties 1 and 2 imply
\begin{equation}\label{error_rkl_1}
    \|\Delta^{(1)}_{R_{k\ell}}\|_2 \le 2\rho_{1k\ell}  \quad \text{and} \quad \|\Delta^{(1)}_{R_{k\ell}}\|_F \le 4\rho_{1k\ell} d_{k\ell}^{1/2},
\end{equation}
and by Property 4, if $R_{k\ell}$ is PSD, then $R_{k\ell}^{(1)}$ is also PSD. 
Combining Corollary \ref{cai_zhang_corollary_smallest_eigval} and \eqref{total_error_bound_first_stage}, we have,  for sufficiently large $n$, 
$$\theta_{k\ell} := \sin_{d_{k\ell}}(R_{k\ell}, R_{k\ell}^{(1)})  \le {\rho_{1k\ell}/2\over b_Rn^\tau /2} = {\rho_{1k\ell}\over b_Rn^\tau},
$$
which implies together with \eqref{total_error_bound_first_stage} that, 
\begin{equation}\label{q_1k_dominates_err_plus_angle}
    \|E_{k\ell}-\Delta_{S+ Q_k}^{(0)}\|_2 + 2\|R_{k\ell}\|_2\theta_{k\ell}^2 \le \rho_{1k\ell}/2 + o(\rho_{1k\ell}) < \rho_{1k\ell}.
\end{equation}
So, by Lemma \ref{avg_errors_lemma} with $s:=s_{u, u}^{(k)}$, $\theta:=\max_{1\le \ell \le m_k}\theta_{k\ell}$, and $\rho := \max_{1\le \ell \le m_k}\rho_{1k\ell}$, we have 
\begin{equation}\label{sum_individ_comp_errors_bound}
    \begin{aligned}
    \Bigl\|{1\over m_k}\sum_{\ell=1}^{m_k} \Delta_{R_{k\ell}}^{(1)} \Bigr\|_2 &\le 11{B_R\over b_R}\max_{1\le\ell\le m_k}\rho_{1k\ell} \Bigl[{1\over m_k} \vee s_{u, u}^{(k)}\vee {1\over b_Rn^{\tau}}\max_{1\le\ell\le m_k}\rho_{1k\ell}\Bigr]^{1/2} .
    \end{aligned}
\end{equation}
Next, we establish an error bound for $\Delta_{S + Q_k}^{(1)}$. With $C_2:={11 (1 + 3\sigma)B_R\over b_R(1\wedge b_R)}$, define
\begin{equation}\label{rho_1k_rate}
    \rho_{1k} =C_2 \max_{1\le\ell\le m_k}\rho_{1k\ell} \Bigl[{1\over m_k} \vee s_{u, u}^{(k)} \vee n^{-\tau}\max_{1\le\ell\le m_k}\rho_{1k\ell}\Bigr]^{1/2} , 
\end{equation}
so that, since $\rho_{1k\ell} \ge 3\sigma(n/m_k)^{1/2}$, 
\begin{equation}\label{total_error_bound_shared_part_first_stage}
    \Bigl\|\bar{E}_k -{1\over m_k}\sum_{\ell=1}^{m_k} \Delta_{R_{k\ell}}^{(1)} \Bigr\|_2 \le 3\sigma (n / m_k)^{1/2} + \Bigl\|{1\over m_k}\sum_{\ell=1}^{m_k} \Delta_{R_{k\ell}}^{(1)} \Bigr\| < \rho_{1k}.
\end{equation}
Then Properties 1 and 2 in Lemma \ref{soft_threshold_properties} imply
\begin{equation}\label{error_sqk_1}
     \|\Delta_{S + Q_k}^{(1)}\|_2 \le 2\rho_{1k} \quad \text{and}\quad \|\Delta_{S + Q_k}^{(1)}\|_F \le 4\rho_{1k}\sqrt{d_0 + d_k}.
\end{equation}

\noindent {\bf Stage II:} We first establish the properties of $Q_{k}^{(1)}, \ k=1, \ldots, K$ by applying Lemma \ref{soft_threshold_properties} with the following  correspondence of notations due to \eqref{second_stage_gaus_with_error}
$$E :=  \bar{E}_{k} - {1\over m_k} \sum_{\ell=1}^{m_k} \Delta_{R_{k\ell}} -\Delta_{S}^{(0)}, \quad  \rho: =\rho_{2k}, \quad d := d_{k}.$$
Define $\rho_{2k} = 4(r^{(II)} \vee \rho_{1k})$, so that by \eqref{total_error_bound_shared_part_first_stage} and the triangular inequality, we have
\begin{equation}\label{total_error_bound_second_stage}
    \|\bar{E}_{k} - {1\over m_k} \sum_{\ell=1}^{m_k} \Delta_{R_{k\ell}} -\Delta_{S}^{(0)}\|_2 \le \|\bar{E}_{k} - {1\over m_k} \sum_{\ell=1}^{m_k} \Delta_{R_{k\ell}}\|_2 +\|\Delta_{S}^{(0)}\|_2 \le \rho_{1k} + r^{(II)} \le \rho_{2k} /2.
\end{equation}
Therefore, by Properties 1 and 2 in Lemma \ref{soft_threshold_properties}, we have 
\begin{equation}\label{error_qk_2}
     \|\Delta_{Q_k}^{(1)}\|_2 \le 2\rho_{2k} \quad \text{and}\quad \|\Delta_{Q_k}^{(1)}\|_F \le 4\rho_{2k}\sqrt{d_k}
\end{equation}
and by Property 4, if $Q_{k}$ is PSD, then $Q_{k}^{(1)}$ is also PSD.
Combining Corollary \ref{cai_zhang_corollary_smallest_eigval} and \eqref{total_error_bound_second_stage}, we get for sufficiently large $n$:
$$\theta_k := \sin_{d_k}(Q_k, Q_k^{(1)})\le {\rho_{2k}\over b_Qn^{\tau}},
$$
which implies together with \eqref{total_error_bound_second_stage}that for sufficiently large $n$ it holds
$$\bigl\|\bar{E}_{k} - {1\over m_k} \sum_{\ell=1}^{m_k} \Delta_{R_{k\ell}} -\Delta_{S}^{(0)}\bigr\|_2  + 2\|Q_{k}\|_2\theta_{k}^2 \le \rho_{2k} / 2 + o(\rho_{2k}) < \rho_{2k}.
$$
Then, Lemma \ref{avg_errors_lemma} with $s:=s_{w, w}$, $\theta := \max_{1\le k\le K} \theta_k$, and $\rho:=\max_{1\le k\le K}\rho_{2k}$ implies
$$\Bigl\|{1\over M}\sum_{k=1}^{K} m_k\Delta_{Q_{k}}^{(1)} \Bigr\|_2 \le  c_K\Bigl\|{1\over K} \sum_{k=1}^{K} \Delta_{Q_{k}}^{(1)} \Bigr\|_2 \le {11c_KB_Q\over b_Q(1\wedge b_Q)}\max_{1\le k\le K}\rho_{2k} \Bigl[{1\over K} \vee s_{w, w} \vee n^{-\tau}\max_{1\le k\le K}\rho_{2k}\Bigr]^{1/ 2},
$$
where $c_K > 0$ is a constant such that $m_k / M \le c_K/K$ for each $k=1, \ldots, K$ (it is guaranteed to exist for sufficiently large $n$ by Assumption \ref{group_props_assump}).
To bound $\Delta_{S}^{(1)}$, we first apply Lemma \ref{avg_errors_lemma} with $s:=s_{u, u}, \theta := \max_{(k, \ell)\in\mathcal{I}}\theta_{k\ell}$, and $\rho:=\max_{(k, \ell)\in\mathcal{I}}\rho_{1k\ell}$ to obtain
\begin{align}\label{avg_R_err_bound}
\Bigl\|{1\over M}\sum_{(k,\ell)\in\mathcal{I}} \Delta_{R_{k\ell}} \Bigr\|_2 & \le 11 {B_R\over b_R} \max_{(k, \ell)\in\mathcal{I}}\rho_{1k\ell}\Bigr[{1\over M} \vee s_{u, u} \vee {1\over b_R n^{\tau}}\max_{(k, \ell)\in\mathcal{I}}\rho_{1k\ell}\Bigl]^{1/2}. 
\end{align}
Letting $C_3 := 12\bigl({c_KB_Q\over b_Q(1\wedge b_Q)} + {B_R\over b_R(1\wedge b_R)}\bigr)$, define
$$\rho_2=C_3\Bigl(\max_{1\le k\le K}\rho_{2k} \bigl[K^{-1} \vee s_{w, w} \vee n^{-\tau}\max_{1\le k\le K}\rho_{2k}\bigr]^{1/2} \vee  \max_{(k, \ell)\in\mathcal{I}}\rho_{1k\ell}\Bigr[{1\over M} \vee s_{u, u} \vee n^{-\tau}\max_{(k, \ell)\in\mathcal{I}}\rho_{1k\ell}\Bigl]^{1/2}\Bigr),
$$
so that by ${\rho_{1k\ell}/M^{1/2}} \ge 3\sigma (n/M)^{1/2} \ge \|\bar{E}\|_2$ we have 
$$\Bigl\|{1\over M}\sum_{(k, \ell)\in\mathcal{I}}\Delta_{R_{k\ell}}^{(1)} \Bigr\|_2 + \Bigl\|{1\over M}\sum_{k=1}^Km_k\Delta_{Q_{k}}^{(1)} \Bigr\|_2 + \|\bar{E}\|_2 <\rho_2 . 
$$
Then, 
for sufficiently large $n$, Properties 1 and 2 in Lemma \ref{soft_threshold_properties} imply
\begin{equation}\label{error_s_2}
       \|\Delta_{S}^{(1)}\|_2 \le 2\rho_{2} \quad \text{and} \quad \|\Delta_{S}^{(1)}\|_F \le 4\rho_{2}\sqrt{d_0}.
\end{equation}
By Property 4, if $S$ is PSD, then $S^{(1)}$ is also PSD.
\end{proof}

We conclude this section by proving Proposition \ref{init_sqrt_n_proposition}, establishing the bounds on the first and second stage averaging initializers. 
\begin{proof}[Proof of Proposition \ref{init_sqrt_n_proposition}]
    We rewrite the first and second stage initializers similarly to \eqref{intuitive_soft_threshold_update} by  substituting soft thresholding with hard thresholding:
    $$(S + Q_k)^{(0)} = \Bigl[{1\over m_{k}}\sum_{\ell=1}^{m_k}A_{k\ell} \Bigr]_{d_0 + d_k}= \Bigl[(S + Q_k) - {1\over m_k}\sum_{\ell=1}^{m_k}R_{k\ell} + \bar{E}_k \Bigr]_{d_0 + d_k},
    $$
    $$S^{(0)} = \Bigl[{1\over K}\sum_{k=1}^K{1\over m_k} \sum_{\ell=1}^{m_k}(A_{k\ell} -\hat{R}_{k\ell})\Bigr]_{d_0}= \Bigl[S + {1\over K} \sum_{k=1}^K Q_k -{1\over K}\sum_{k=1}^K{1\over m_k}\sum_{\ell=1}^{m_k}\Delta_{R_{k\ell}} + \bar{E}_k \Bigr]_{d_0 + d_k}.
    $$ 
   Our next goal is to use Lemma \ref{hard_threshold_lemma} to establish deterministic bounds on the errors of two initializers, assuming $\mathcal{E}_{noise}$ holds. 
    \\
    \\
    {\bf Stage I.} 
    By Lemma 4 in \citep{multiness}, we have 
    $$\|{1\over m_k}\sum_{\ell=1}^{m_k}R_{k\ell}\|_2 \le B_Rn^{\tau}(m_k^{-1} + s_{u, u}^{(k)})^{1/2} 
    $$
   Therefore, by triangular inequality, it holds for $E:=\bar{E}_k -  {1\over m_k}\sum_{\ell=1}^{m_k}R_{k\ell}$:
    $$\|E\|_2 \le 3\sigma (n / m_k)^{1/2} + B_Rn^\tau\bigl(m_k^{-1} + s_{u, u}^{(k)}\bigr)^{1/2} = o(n^\tau)$$
    and by Property 2 in Lemma \ref{hard_threshold_lemma}, we obtain for sufficiently large $n$
    $$\|(S + Q_k)  - (S + Q_k)^{(0)}\|_2   \le {19B_{S + Q} \over b_{S +Q}}\|\mathcal{P}_{[V, W_k]} E\|_2.
    $$
    The needed result follows by triangular inequality and Lemma 4 in \citep{multiness}:
    \begin{align*}
    \|\mathcal{P}_{[V, W_k]} E\|_2& \le \|\bar{E}_k\|_2 + \|\sum_{\ell=1}^{m_k}\mathcal{P}_{[V, W_k]} R_{k\ell}\|_2 /m_k\\
    &\le 3\sigma(n/m_k)^{1/2} + B_Rn^\tau s_{vw, u}^{(k)}(m_k^{-1} + s_{u,u}^{(k)})^{1/2} \lesssim n^{1/2}.
    \end{align*}  
\\
    {\bf Stage II.}  
    Combining \eqref{avg_R_err_bound} with the  established rates $r^{(I)}_k \lesssim n^{1/2}$ for the first-stage initializers  above, we obtain for sufficiently large $n$,
    $$\Bigl\|{1\over K}\sum_{k=1}^K{1\over m_k}\sum_{\ell=1}^{m_k}\Delta_{R_{k\ell}} \Bigl\|_2 \le c_K'\Bigl\|{1\over M}\sum_{(k,\ell)\in\mathcal{I}} \Delta_{R_{k\ell}} \Bigr\|_2 \lesssim (n/M)^{1/2},
    $$
    where $c_K' > 0$ is a constant such that $M / K\le c_K'm_k$ for every $k=1, \ldots, K$ (it is guaranteed to exist by Assumption \ref{group_props_assump}).
    By Lemma 4 in \citep{multiness}, we also have 
    $$\Bigl\|{1\over K} \sum_{k=1}^K Q_k\Bigr\|_2 \le B_Q n^\tau (K^{-1} + s_{w, w})^{1/2}  
    $$
    Therefore, by triangular inequality, it holds for $E:={1\over K} \sum_{k=1}^K Q_k + \bar{E} - {1\over K}\sum_{k=1}^K{1\over m_k}\sum_{\ell=1}^{m_k}\Delta_{R_{k\ell}}$ if $n$ is sufficiently large
    $$\|E\|_2 \le (1 + \delta) B_Q n^\tau (K^{-1} + s_{w, w})^{1/2}.
    $$
    This implies $|\gamma_{d_0}(S)| \ge b_{S} n^\tau \ge 4\|E\|_2$ by \eqref{min_S_max_Q_sep}. So, we can apply Property 2 in Lemma \ref{hard_threshold_lemma} to obtain:
    $$\|S^{(0)} - S\|_2 \le {19B_{S} \over b_S}\|\mathcal{P}_V E\|_2.
    $$
    Finally, Lemma 4 in \citep{multiness} and triangular inequality imply the needed rate
    \begin{align*}
        \|\mathcal{P}_V E\|_2 &\le \|\sum_{k=1}^K \mathcal{P}_V Q_k\|_2 / K + O(\sqrt{n/M})  \\
        &\le B_Qn^\tau s_{v, w}(K^{-1} + s_{w, w})^{1/2} +  O(\sqrt{n/M}) \lesssim \sqrt{nK / M}.
        \end{align*}
    %
\end{proof}

\subsection{Technical lemmas}\label{technical_lemmas_section}
In this section, we present the proofs of the technical lemmas used to establish the main result of Theorem \ref{main_consistency_theorem} and Proposition \ref{init_sqrt_n_proposition}.

\begin{lemma}[Version of Theorem 1, \citep{cai_zhang}]\label{version_cai_zhang_thm}
Consider symmetric matrices $Z, E \in \mathbb{R}^{n\times n}$. Define the eigen-decompositions of $Z$ and its perturbation $\tilde{Z} = Z + E$ as, respectively, 
\begin{align*} 
Z =\left[\begin{array}{ll}
V & V_{\perp}
\end{array}\right] \left[\begin{array}{cc}
\Gamma_1 & 0 \\
0 & \Gamma_2
\end{array}\right] \left[\begin{array}{c}
V^{\top} \\
V_{\perp}^{\top}
\end{array}\right] , \ \ \ 
\tilde{Z}  =\left[\begin{array}{ll}
\tilde{V} & \tilde{V}_{\perp}
\end{array}\right] \left[\begin{array}{cc}
\tilde{\Gamma}_1 & 0 \\
0 & \tilde{\Gamma}_2
\end{array}\right] \left[\begin{array}{c}
\tilde{V}^{\top} \\
\tilde{V}_{\perp}^{\top}
\end{array}\right],
\end{align*}
where $V^\top V = I_d, \ V_\perp^\top V_\perp = I_{n-d}, \ \Gamma_1 = \operatorname{diag}[\gamma_1(Z), \ldots, \gamma_d(Z)], \ \Gamma_2 = \operatorname{diag}[\gamma_{d+1}(Z), \ldots, \gamma_n(Z)]$ for some $d < n$, and $\tilde{\Gamma}_1, \tilde{\Gamma}_2, \tilde{V}, \tilde{V}_{\perp}$ are defined similarly. Let  
$$\alpha = |\gamma_{\min}(V^{\top} \tilde{Z}V)|,\ \beta = \|V_\perp^\top \tilde{Z}V_\perp\|_2, \ e_{21}=\|\mathcal{P}_{V_{\perp}} E \mathcal{P}_V\|_2, \ e_{12}=\|\mathcal{P}_V E \mathcal{P}_{V_{\perp}} \|_2.
$$ Then, if 
\begin{equation}\label{cai_zhang_reformulated_condition}
    \alpha > \beta + e_{12}\wedge e_{21},
\end{equation}
it holds
\begin{equation}
    \|\sin \Theta(V, \hat{V})\|_2 \leq \frac{ e_{12}\vee e_{21}}{\alpha-\beta-e_{21} \wedge e_{12}} 
\end{equation}
\end{lemma}
\begin{proof}
Condition \eqref{cai_zhang_reformulated_condition} is stronger than the one in Theorem 1, \citep{cai_zhang} since it implies
$$ \alpha^2 > (\beta + e_{12}\wedge e_{21})^2 \ge \beta^2 + e_{12}^2 \wedge e_{21}^2,$$
and we can further relax the original upper bound on $\|\sin \Theta(V, \hat{V})\|_2$ as follows: 
\begin{align*}
    \|\sin \Theta(V, \hat{V})\|_2 &  \le \ {\alpha e_{12} + \beta e_{21}\over \alpha^2 -\beta^2 - e_{12}^2\wedge e_{21}^2}  
    \le \ {(\alpha + \beta + e_{12}\wedge e_{21}) (e_{12} \vee e_{21})\over \alpha^2 -(\beta + e_{12} \wedge  e_{21})^2} \\
    & =  \ {e_{12} \vee e_{21} \over \alpha - \beta - e_{12}\wedge e_{21}}.
\end{align*}
\end{proof}
\begin{corollary}  \label{cai_zhang_corollary_smallest_eigval}  
     In notation of Lemma \ref{version_cai_zhang_thm}, if $d$ is the rank of $Z$ and $|\gamma_d(Z)| > 3\|E\|_2$, then
        $$\|\sin\Theta(V, \tilde{V})\|_2 \le {\|\mathcal{P}_VE\|_2 \over |\gamma_d(Z)| - 3\|E\|_2 }.
        $$
\end{corollary}
\begin{proof}
     By $e_{12}, e_{21} \le \|\mathcal{P}_VE\|_2 $ and  Weyl's inequality combined with submultiplicativity,
\begin{align*} 
\alpha &= |\gamma_{\min} (V^\top\tilde{Z} V)| \ge |\gamma_d(V^\top ZV + V^\top EV)| \ge |\gamma_d(Z)| -\|E\|_2 , \\
    \beta &= \|V_\perp^\top \tilde{Z} V_\perp\|_2 \le \|V_\perp^\top ZV_\perp\|_2 + \|V_\perp^\top EV_\perp\|_2 \le \|E\|_2.
\end{align*}
    Then $\alpha - \beta -e_{12}\wedge e_{21} \ge |\gamma_d(Z)| - 3\|E\|_2 > 0$ by our assumption, and the needed bound follows by Lemma \ref{version_cai_zhang_thm}.
\end{proof}

The next lemma lists various relationships between a matrix and the soft thresholding of its perturbation.
\begin{lemma}[Soft thresholding with noise]\label{soft_threshold_properties}
    Consider symmetric matrices $Z, E \in \mathbb{R}^{n\times n}$. Define the perturbation $\tilde{Z} = Z + E$ and its soft thresholding $\hat{Z} = \mathcal{T}_\rho(\tilde{Z})$ for some $\rho > 0$. Let $d = \mathrm{rank}(Z)$ and define $V, \tilde{V} \in\mathbb{R}^{n\times d}$ as the top $d$ eigenvectors of $Z$ and $\tilde{Z}$, respectively. Then the following hold: 
    \begin{enumerate}
        \item The spectral norm of the difference can be bounded as
        $$\|\hat{Z} - Z\|_2 \le \rho + \|E\|_2.
        $$
        \item If $\rho \ge \|E\|_2$, the Frobenius norm of the difference satisfies 
        $$\|\hat{Z}-Z\|_F \leq 4 \rho \sqrt{d}$$ 
        \item If $\rho \ge \|E\|_2$,  $\|\hat{Z}\|_2 \le \|Z\|_2$,
        \item If $\rho \ge \|E\|_2$ and $Z$ is PSD, then $\hat{Z}$ is also PSD.
        \item If $\rho \ge \|E\|_2 + 2\|Z\|_2\|\sin \Theta(V, \tilde{V})\|_2^2$,
        then $\mathrm{rank}(\hat{Z}) \le d$.
    \end{enumerate}
\end{lemma}

\begin{proof}
\begin{enumerate}
    \item Consider the eigendecomposition $\tilde{Z} =Z + E = \tilde{U}\tilde{\Gamma}\tilde{U}^\top$. Then, the eigendecomposition of $\hat{Z}$ can be written as $\hat{Z} = \tilde{U}\tilde{\Gamma}_\rho \tilde{U}^\top$  with 
    $$\tilde{\Gamma}_\rho = \operatorname{diag}\left[\operatorname{sign}(\tilde{\Gamma}_{11})(|\tilde{\Gamma}_{11}| - \rho)_+, \ldots, \operatorname{sign}(\tilde{\Gamma}_{nn})(|\tilde{\Gamma}_{nn}| - \rho)_+ \right].$$ So, we can decompose the error as
    $$\hat{Z}-Z=\tilde{U} \tilde{\Gamma}_\rho \tilde{U}^\top-\left(\tilde{U} \tilde{\Gamma} \tilde{U}^\top -E\right)=\tilde{U}\left(\tilde{\Gamma}_\rho-\tilde{\Gamma}\right) \tilde{U}^\top+E$$
    where $\tilde{\Gamma}_\rho-\tilde{\Gamma}$ is diagonal with
    $$[\tilde{\Gamma}_\rho-\tilde{\Gamma}]_{ii} = \operatorname{sign}(\tilde{\Gamma}_{ii})(|\tilde{\Gamma}_{ii}| - \rho)_+ - \tilde{\Gamma}_{ii} = -\operatorname{sign}(\tilde{\Gamma}_{ii})\min(|\tilde{\Gamma}_{ii}|, \rho), 
    $$
    From that, the needed bound on the spectral norm of the error follows immediately
    $$\|\hat{Z}-Z\|_2 \le \|\tilde{\Gamma}_\rho-\tilde{\Gamma}\|_2 + \|E\|_2 = \max_{1\le i \le n} |\min\{\rho, |\tilde{\Gamma}_{ii}|\}| + \|E\|_2 \le \rho + \|E\|_2.$$

    \item Consider the variational definition of soft thresholding:
    $$\hat{Z} = \argmin_{Y\in\mathbb{R}^{n\times n}} \Bigl\{{1\over 2}\|Y - (Z + E)\|_F^2 + \rho \|Y\|_*\Bigr\}
    $$
    By optimality, there exists a subgradient $G_Z \in \partial \|\hat{Z}\|_*$ such that 
        \begin{equation}\label{optim_of_Z_hat}
             \langle\hat{Z} - Z - E +\rho G_Z, \hat{Z}-\check{Z}\rangle \leq 0
        \end{equation}
    for any matrix $\check{Z}$ with the same column space and row space as $Z$. Let \( \check{G}_Z \in \partial\|\check{Z}\|_* \) be arbitrary. Adding and subtracting \( \rho\langle\check{G}_Z, \hat{Z}-\check{Z}\rangle \) gives  
    \begin{equation}\label{optim_with_arbitrary_subgrad}
        \langle \hat{Z}-Z, \hat{Z}-\check{Z}\rangle+\rho\langle G_Z-\check{G}_Z, \hat{Z}-\check{Z}\rangle \leq\langle-\rho \check{G}_Z+E, \hat{Z}-\check{Z}\rangle
    \end{equation}
    On the other hand, by convexity of the nuclear norm, we have
    $$
    \begin{aligned}
    \|\hat{Z}\|_*-\|\check{Z}\|_* \geq\left\langle\check{G}_Z, \hat{Z}-\check{Z}\right\rangle ,  \  \ 
   \|\check{Z}\|_*- \|\hat{Z}\|_* \geq\left\langle G_Z, \check{Z}-\hat{Z}\right\rangle
    \end{aligned}
    $$
    which together imply 
    \begin{equation}\label{convex_nuclear_norm_optim}
        \langle G_Z-\check{G}_Z, \hat{Z}-\check{Z}\rangle \geq 0.
    \end{equation}
    Combining \eqref{optim_with_arbitrary_subgrad} and \eqref{convex_nuclear_norm_optim}, we obtain
    \begin{equation}\label{final_before_bounding}
        \langle \hat{Z}-Z, \hat{Z}-\check{Z}\rangle \le -\rho\langle \check{G}_Z, \hat{Z}-\check{Z}\rangle + \langle E, \hat{Z}-\check{Z}\rangle , 
    \end{equation}
    which is the inequality of the same form as \eqref{optim_of_Z_hat} but now instead of a fixed $G_Z$ we have a free choice of $\check{G}_Z$. 
    By \cite{Koltchinskii2010}, subgradient $\check{G}_Z$ can be expressed as 
    $$
    \check{G}_Z=\sum_{i=1}^du_iv_i^\top+\check{\mathcal{P}}_{Z}^{\perp} W \check{\mathcal{P}}_{Z}^{\perp} , 
    $$
    where $\check{\mathcal{P}}_{Z}$ is the projector on the column space of $\check{Z}$, vectors $u_i, v_i$ are left and right singular vectors of $\check{Z}$, respectively, and $W$ is an arbitrary matrix satisfying \( \|W\|_2 \leq 1 \). We can specify \( W \) to attain the upper bound of the following expression obtained using the duality of nuclear and operator norms
    \begin{equation}\label{bound_proj_W_proj}
    \langle\check{\mathcal{P}}_{Z}^{\perp} W \check{\mathcal{P}}_Z^\perp, \hat{Z}-\check{Z}\rangle=\langle W, \check{\mathcal{P}}_{Z}^{\perp} \hat{Z} \check{\mathcal{P}}_{Z}^{\perp}\rangle \le \|\check{\mathcal{P}}_{Z}^{\perp} \hat{Z} \check{\mathcal{P}}_{Z}^{\perp}\|_{*}.
    \end{equation}
    By trace duality, we can also bound
    \begin{equation}\label{bound_prod_sum_uv_diff}
    |\langle \sum_{i=1}^d u_i v_i^{\top}, \hat{Z}-\check{Z}\rangle | \leq \sqrt{d} \|\hat{Z}-\check{Z}\|_F . 
    \end{equation}
    Finally, we bound the second term in the RHS of \eqref{final_before_bounding} again by trace duality
    \begin{equation}\label{bound_prod_error_diff}
    \begin{aligned}
    \langle E, \hat{Z}-\check{Z}\rangle & =\langle E-\check{\mathcal{P}}_{Z}^{\perp} E \check{\mathcal{P}}_{Z}^{\perp}, \hat{Z}-\check{Z}\rangle+\langle\check{\mathcal{P}}_{Z}^{\perp} E \check{\mathcal{P}}_{Z}^{\perp}, \hat{Z}-\check{Z}\rangle \\
    & = \langle \check{\mathcal{P}}_{Z}^{\perp} E \check{\mathcal{P}}_{Z} + \check{\mathcal{P}}_{Z} E \check{\mathcal{P}}_{Z}^{\perp}  + \check{\mathcal{P}}_{Z} E \check{\mathcal{P}}_{Z}, \hat{Z}-\check{Z}\rangle + \langle E, \check{\mathcal{P}}_{Z}^{\perp} \hat{Z}\check{\mathcal{P}}_{Z}^{\perp} \rangle\\
    & \leq 3 \sqrt{d}\left\|E\right\|_2\|\hat{Z}-\check{Z}\|_F+\left\|E\right\|_2\|\check{\mathcal{P}}_{Z}^{\perp} \hat{Z} \check{\mathcal{P}}_{Z}^{\perp}\|_*.
    \end{aligned}
    \end{equation}
    Plugging the bounds \eqref{bound_proj_W_proj}, \eqref{bound_prod_sum_uv_diff}, and \eqref{bound_prod_error_diff} into \eqref{final_before_bounding} and specifying \( \check{Z} = Z \), we get
    \[
    \begin{aligned}
    \|\hat{Z} - Z\|_F^2 & +\left(\rho-\left\|E\right\|_2\right)\|\check{\mathcal{P}}_{Z}^{\perp} \hat{Z} \check{\mathcal{P}}_{Z}^{\perp}\|_* \leq \sqrt{d}\left(\rho +3 \left\|E\right\|_2\right)\|\hat{Z} - Z\|_F . 
    \end{aligned}
    \]
    Using the assumption $\rho \ge \|E\|_2$ and dividing both sides by $\|\hat{Z} - Z\|_F$, we deduce the needed bound
    $$\|\hat{Z} - Z\|_F \le \sqrt{d}\left(\rho +3 \left\|E\right\|_2\right) \le 4\sqrt{d}\rho.
    $$
    \item By definition of soft-thresholding and triangular inequality, we have  
    $$\|\hat{Z}\|_2 \le \|\tilde{Z}\|_2 - \rho \le \|Z\|_2 + \|E\|_2 - \rho \le\|Z\|_2 . 
    $$
    \item If $Z$ is PSD, it holds $v^\top Zv \ge 0$ for any $v\in\mathbb{R}^n$ and we can bound
    $$v^\top(Z + E)v \ge v^\top Ev \ge -\|E\|_2  \ge -\rho,
    $$
    so only positive eigenvalues can survive after applying soft-thresholding $\mathcal{T}_\rho$  to $Z + E$. 
    \item Let $v$ be a unit vector from the orthogonal complement of $\operatorname{col}(\tilde{V})$. Then, by orthogonality and \eqref{infO_sin_theta}
    \begin{equation*}
    \left\|V^{\top} v\right\|_2 = \inf _{O \in \mathcal{O}_{d}}
    \left\|V^{\top} v-O \tilde{V}^{\top} v\right\|_2 
    \leq \inf _{O \in \mathcal{O}_{d}}
    \left\|V-O \tilde{V}\right\|_2 \le \sqrt{2}\|\sin\Theta(V, \tilde{V})\|_2
    \end{equation*}
    and
    \begin{equation}
    \left|v^{\top} Z v\right|=
    \left|v^{\top} V V^{\top} Z V V^{\top} v\right| 
    \leq\left\|Z\right\|_2 \left\|V^{\top} v\right\|_2^2 . 
    \end{equation}
    Therefore, 
    \begin{equation*}
    \left|v^{\top}\left(Z+E\right) v\right| 
    \leq \left|v^{\top} Z v\right|+\left\|E\right\|_2 \le \rho.
    \end{equation*} 
    Thus at most first $d$ eigenvalues of $\hat{Z}$ survive the soft-thresholding step
    $$
    \left|\gamma_{d+1}\left(Z+E\right)\right|= 
    \max_{v \perp \operatorname{col}\left(\tilde{V}\right)} 
    \left|v^{\top}\left(Z+E\right) v\right| \le \rho.
    $$
\end{enumerate}
\end{proof}

Finally, in Lemma \ref{avg_errors_lemma}, we establish a generic bound on the spectral norm of the average difference between matrices $Z_\ell, \ell=1, \ldots, m$ and their soft thresholding estimators $\hat{Z}_\ell$ defined similarly to Lemma \ref{soft_threshold_properties}.

\begin{lemma}[Average of soft-thresholding estimators] \label{avg_errors_lemma}
Consider symmetric matrices $Z_\ell, E_\ell \in \mathbb{R}^{n\times n}, \ell = 1, \ldots, m$ and define the perturbations $\tilde{Z}_\ell = Z_\ell + E_\ell$. Let $d_\ell$ be the rank of $Z_\ell$, and define $V_\ell, \tilde{V}_\ell \in\mathbb{R}^{n\times d_\ell}$ as the top $d_\ell$ eigenvectors of $Z_\ell$ and $\tilde{Z}_\ell$, respectively.
Consider estimators $\hat{Z}_\ell = \mathcal{T}_{\rho_\ell}(\tilde{Z}_\ell)$ with thresholds $\rho_\ell$ satisfying 
\begin{equation}\label{threshold_cond_avg_error}
    \rho_\ell > \|E_\ell\|_2 + 2\|Z_\ell\|_2\|\sin \Theta(V_\ell, \tilde{V}_\ell)\|_2^2.
\end{equation}
Define also
\begin{align*}
\rho = \max_{1\le \ell\le m} \rho_\ell, \quad \theta =& \max_{1\le \ell \le m} \|\sin \Theta(V_\ell, \tilde{V}_\ell)\|_2 , \quad s = \max_{1\le \ell_1 < \ell_2 \le m}\|V_{\ell_1}^\top V_{\ell_2}\|_2, \quad \gamma = \max_{1\le \ell \le m} \|Z_\ell\|_2.
\end{align*}
Then 
\begin{align}
    \Bigl\|{1\over m}\sum_{\ell=1}^m (\hat{Z}_\ell - Z_\ell)\Bigr\|_2 \le  11(\gamma\theta \vee \rho) \Bigl[{1\over m} \vee s \vee \theta\Bigr]^{1/2} . 
\end{align}
    
\end{lemma}
\begin{proof}
With the choice of the thresholds in \eqref{threshold_cond_avg_error}, we   deduce by property 5 in Lemma \ref{soft_threshold_properties} that $\operatorname{rank}(\hat{Z}_\ell) \le d_\ell$. Therefore, $\hat{Z}_\ell$ satisfies
\begin{equation*}
\hat{Z}_\ell=\tilde{V}_\ell \tilde{V}_\ell^{\top} \hat{Z}_\ell \tilde{V}_\ell \tilde{V}_\ell^{\top}.
\end{equation*}
Thus, we can decompose  $\Delta_\ell:=\hat{Z}_\ell - Z_\ell$ as 

\begin{equation*}
\begin{aligned}
\Delta_\ell &= \tilde{V}_\ell \tilde{V}_\ell^{\top} \hat{Z}_\ell \tilde{V}_\ell \tilde{V}_\ell^{\top}-V_\ell V_\ell^{\top} Z_\ell V_\ell V_\ell^{\top} \\
& =\left(\tilde{V}_\ell \tilde{V}_\ell^{\top}-V_\ell V_\ell^{\top}\right) \hat{Z}_\ell \tilde{V}_\ell \tilde{V}_\ell^{\top}+V_\ell V_\ell^{\top} \Delta_\ell V_\ell V_\ell^{\top}+V_\ell V_\ell^{\top} \hat{Z}_\ell\left(\tilde{V}_\ell \tilde{V}_\ell^{\top}-V_\ell V_\ell^{\top}\right) \\
& =:I_\ell +I I_\ell + I I I_\ell .
\end{aligned}
\end{equation*}
We bound the operator norm of the sum over $\ell$ for each of these three terms separately. \\
\\
{\it Term I}. 
We start by bounding the cosine similarity between perturbed eigenspaces in terms of the similarity of the true eigenspaces. By \eqref{infO_sin_theta} and the triangular inequality, we have
\begin{equation}\label{vvT_tilde_cos_similarity}
\begin{aligned}
\left\|\tilde{V}_{\ell_1}^{\top} \tilde{V}_{\ell_2}\right\|_2 & =\inf_{O_1\in \mathcal{O}_{d_{\ell_1}},\  O_2\in \mathcal{O}_{d_{\ell_2}}}\left\|\tilde{V}_{\ell_1}^{\top} \tilde{V}_{\ell_2}-\tilde{V}_{\ell_1}^{\top} V_{\ell_2} O_2+\tilde{V}_{\ell_1}^{\top} V_{\ell_2} O_2-O_1 V_{\ell_1}^{\top} V_{\ell_2} O_2+O_1 V_{\ell_1}^{\top} V_{\ell_2} O_2\right\|_2 \\
& \leq\inf_{O_2\in \mathcal{O}_{d_{\ell_2}}}\left\|\tilde{V}_{\ell_2}-V_{\ell_2} O_2\right\|_2 +\inf_{O_1\in \mathcal{O}_{d_{\ell_1}}}\left\|\tilde{V}_{\ell_1}-V_{\ell_1} O_1\right\|_2+\left\|V_{\ell_1}^{\top} V_{\ell_2}\right\|_2 \\
& \le 2\sqrt{2}\theta + s . 
\end{aligned}
\end{equation}
By property 3 in Lemma \ref{soft_threshold_properties} and the threshold choice in \eqref{threshold_cond_avg_error}, we have $\|\hat{Z}_\ell\|_2 \le \gamma$.
Thus, by \eqref{vvT_sin_theta} and submultiplicativity
\begin{equation*}
\begin{aligned}
    \left\|\left(\tilde{V}_\ell \tilde{V}_\ell^{\top}-V_\ell V_\ell^{\top}\right) \hat{Z}_\ell \tilde{V}_\ell \tilde{V}_\ell^{\top}\right\|_2 & \leq \|\tilde{V}_\ell \tilde{V}_\ell^{\top}-V_\ell V_\ell^{\top}\|_2 \|\hat{Z}_\ell\|_2 \leq 2\theta \gamma . 
\end{aligned}
\end{equation*}
Similarly,  by \eqref{vvT_tilde_cos_similarity}
\begin{equation} \label{term1_cross_prod}
\begin{aligned}
    \left\|\left(\tilde{V}_{\ell_1} \tilde{V}_{\ell_1}^{\top}-V_{\ell_1} V_{\ell_1}^{\top}\right) \hat{Z}_{\ell_1} \tilde{V}_\ell \tilde{V}_{\ell_1}^{\top} \cdot \tilde{V}_{\ell_2} \tilde{V}_{\ell_2}^{\top}\hat{Z}_{\ell_2}\left(\tilde{V}_{\ell_2} \tilde{V}_{\ell_2}^{\top}-V_{\ell_2} V_{\ell_2}^{\top}\right)\right\|_2 \le \\
    \|\tilde{V}_{\ell_1}^{\top} \tilde{V}_{\ell_2}\|_2 \cdot \max_\ell \|\hat{Z}_{\ell}\|_2^2 \cdot \max_{\ell_1 < \ell_2} \|\tilde{V}_{\ell_1} \tilde{V}_{\ell_1}^{\top}-V_{\ell_1} V_{\ell_1}^{\top}\|_2^2 \ \le\ (2\sqrt{2}\theta + s) \cdot 4\theta^2\gamma^2 .
\end{aligned}
\end{equation}
Therefore, by Lemma 4 in \cite{multiness}
$$\Bigl\|{1\over m}\sum_{\ell=1}^m I_\ell\Bigr\|_2 \le 2m^{-1/2}\theta\gamma \Bigl[1 + m(2\sqrt{2}\theta + s) \Bigr]^{1/2} \le 5\theta\gamma\Bigl[{1\over m} \vee\theta\vee s \Bigr]^{1/2},
$$
where we used $(2 + 2\sqrt{2})^{1/2} < 2.5$.\\
\\
{\it Term II.} 
With the choice of the thresholds in \eqref{soft_threshold_properties}, we have by property 1 in Lemma \ref{soft_threshold_properties} that $\|\Delta_\ell\|_2 \le 2\rho$
\begin{align*}
    & \|V_{\ell_1} V_{\ell_1}^{\top} \Delta_{\ell_1}  V_{\ell_1} V_{\ell_1}^{\top} \cdot V_{\ell_2}  V_{\ell_2}^{\top} \Delta_{\ell_2}  V_{\ell_2}   V_{\ell_2}^{\top}\|_2 \le 
    \|\Delta_{\ell_1}\|_2\|V_{\ell_1}^{\top} V_{\ell_2}\|_2 \|\Delta_{\ell_2}\|_2 \le 4s\rho^2 . 
\end{align*}
Thus, by Lemma 4 in \cite{multiness}, 
$$\Bigl\|{1\over m}\sum_{\ell=1}^m II_\ell\Bigr\|_2 \le 2m^{-1/2}\rho\Bigl[1 + ms\Bigr]^{1/2} \le 3\rho\Bigl[{1\over m} \vee s\Bigr]^{1/2} . 
$$
{\it Term III.} Similarly to Term $I$, we can establish by substituting $\|\tilde{V}_{\ell_1}\tilde{V}_{\ell_2}\|_2$ with $\|V_{\ell_1}V_{\ell_2}\|_2$ in \eqref{term1_cross_prod}:
$$\Bigl\|{1\over m}\sum_{\ell=1}^m III_\ell\Bigr\|_2 \le 2m^{-1/2}\theta\gamma \Bigl[1 + ms \Bigr]^{1/2} \le 3\theta\gamma\Bigl[{1\over m}\vee s \Bigr]^{1/2}.
$$
\end{proof}

\begin{lemma}[Hard thresholding with noise]\label{hard_threshold_lemma}
In notation of Lemma \ref{soft_threshold_properties} with  $\hat{Z} := [\tilde{Z}]_d$
\begin{enumerate}
    \item The spectral norm of the difference can be bounded as 
    $$\|\hat{Z} - Z\|_2 \le 2\|E\|_2
    $$
    \item If additionally $|\gamma_d(Z)| \ge 4\|E\|_2$, it holds
    $$\|\hat{Z} - Z\|_2 \le {19\|Z\|_2 \over |\gamma_d(Z)|}\|\mathcal{P}_{V}E\|_2.$$
\end{enumerate}
\end{lemma}

\begin{proof} 
\begin{enumerate}
    \item Consider the eigendecomposition $Z + E = \tilde{U}\tilde{\Gamma}\tilde{U}^\top$. Then, the eigendecomposition of $\hat{Z}$ can be written as $\hat{Z} = \tilde{U}[\tilde{\Gamma}]_d \tilde{U}^\top$  with 
    $$\hat{Z}-Z=\tilde{U} [\tilde{\Gamma}]_d\tilde{U}^\top-\left(\tilde{U} \tilde{\Gamma} \tilde{U}^\top -E\right)=\tilde{U}\left([\tilde{\Gamma}]_d -\tilde{\Gamma}\right) \tilde{U}^\top+E$$
    Therefore, by Weil's inequality
    $$\|\hat{Z}-Z\|_2 \le |\gamma_{d+1}(Z + E)| + \|E\|_2 \le |\gamma_{d+1}(Z)| + 2\|E\|_2 = 2\|E\|_2$$
    \item  Decompose the error as follows:
        \begin{align*}
       \hat{Z} - Z &= \mathcal{P}_{\tilde{V}} \tilde{Z}\mathcal{P}_{\tilde{V}}  - \mathcal{P}_{V}  Z\mathcal{P}_{V} = (\mathcal{P}_{\tilde{V}} - \mathcal{P}_{{V}} ) \tilde{Z} \mathcal{P}_{\tilde{V}}  +\mathcal{P}_{V}  E\mathcal{P}_{\tilde{V}} +\mathcal{P}_{V} Z(\mathcal{P}_{\tilde{V}} - \mathcal{P}_{V})
    \end{align*}
    By Corollary \ref{cai_zhang_corollary_smallest_eigval} and \eqref{vvT_sin_theta},
    $$\|\mathcal{P}_{V} - \mathcal{P}_{\tilde{V}}\|_2 = \|\tilde{V}\tilde{V}^\top - V^\top V^\top\|_2 \le 2\sin\Theta(V, \tilde{V}) \le {2\|\mathcal{P}_{V}E\|_2\over |\gamma_d(Z)| - 3\|E\|_2}
    $$
    Notice also that $\|\tilde{Z}\|_2 \le \|Z\|_2 + \|E\|_2 \le \|Z\|_2 + |\gamma_d(Z)| / 4 = 5\|Z\|_2 /4 $.
    Therefore, by triangular inequality and submultiplicativity:
    \begin{align*}
        \| \hat{Z} - Z\|_2 &\le \|\mathcal{P}_{V}E\|_2 + (\|Z\|_2 + \|\tilde{Z}\|_2) {2\|\mathcal{P}_{V}E\|_2\over |\gamma_d(Z)| - 3\|E\|_2}  \\
        & \le \|\mathcal{P}_{V} E\|_2 + {9\|Z\|_2 \over  4} {2 \|\mathcal{P}_{V}E\|_2\over |\gamma_d(Z)| / 4} \\
        &\le {19\|Z\|_2 \over |\gamma_d(Z)|}\|\mathcal{P}_{V}E\|_2.
    \end{align*}
\end{enumerate}
\end{proof}


\section{Additional experiments}

\subsection{Dependency of estimation errors on the cosine similarities of the components}
\label{cosine_similarity_appendix_section}

In this section, we present additional experiments exploring how the latent component estimation error depends on their cosine similarities.
Similarly to Section \ref{sec:experiments}, we generate latent components by Algorithm \ref{alg:sampling_procedure} with $n=200$, $M=16$, $K=4$, $s_{v,u} = s_{w, u} = 0.1$ and vary one of the cosine similarities $s_{w,w}, s_{u,u}$ over the grid $[0.1, 0.2, \ldots, 0.9]$  while keeping the others fixed at 0.1.  Figure \ref{fig:error_angle_dependency} presents the results for the Gaussian edge distribution;  the logistic link results are very similar.

\begin{figure}[th!]
    \centering
    \includegraphics[width=0.9\linewidth]{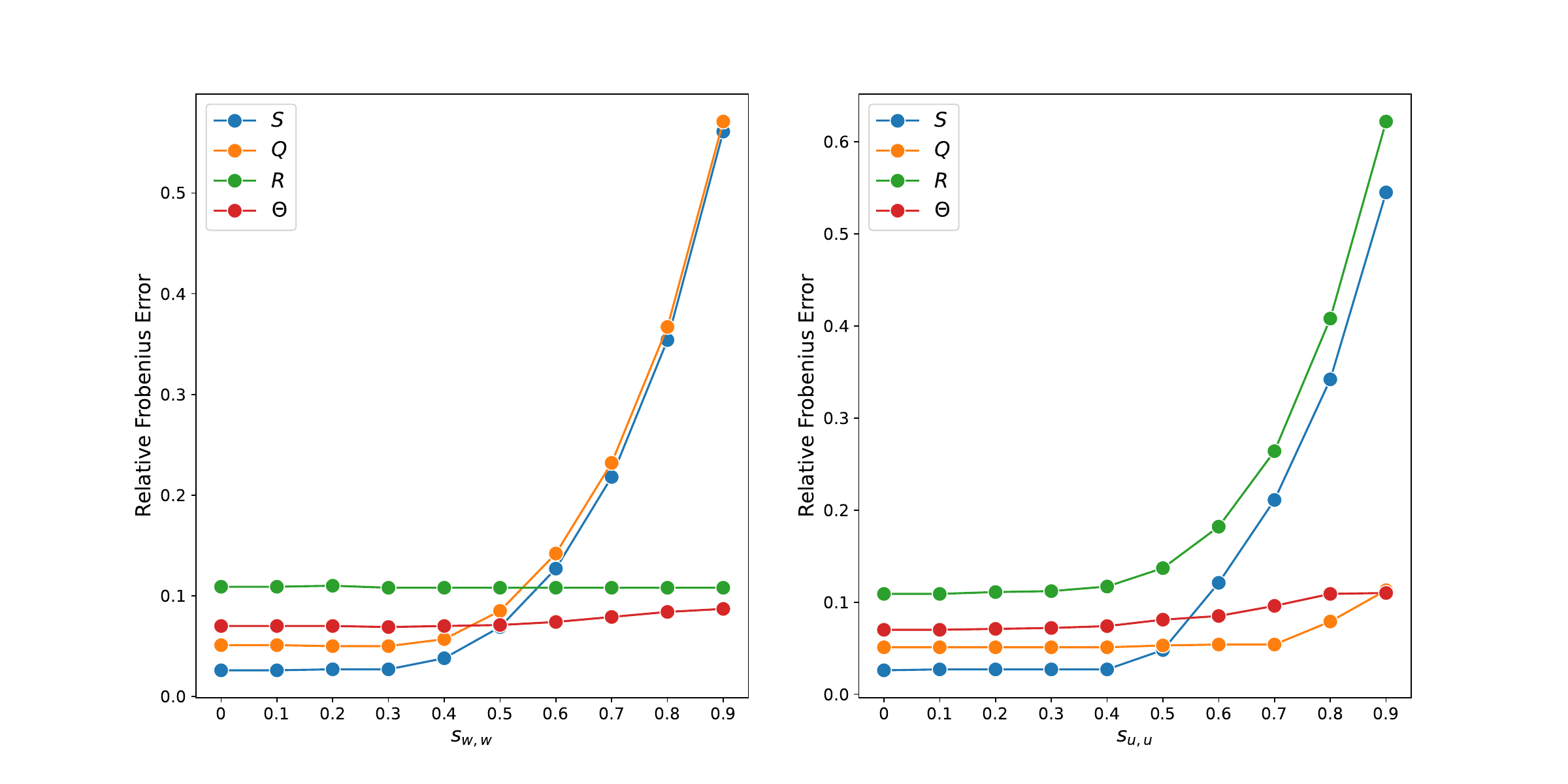}
    \caption{Dependency of ARFE on the cosine similarities between group and individual components. Unless one of the parameters is varied, the networks are generated according to Algorithm \ref{alg:sampling_procedure} with $M=16$ layers on $n=200$ nodes, $K=4$ balanced groups, and $s_{v, u}=s_{w, u} = 0.1$.}
    \label{fig:error_angle_dependency}
\end{figure}

In general, we see that an increase in any of the similarities has a very small effect on the error of $\Theta$. We conjecture that higher similarity of latent components primarily makes their separation harder but has less influence on the effective sample size for estimating their sum.

In the $s_{w, w}$ plot, we can observe that the similarity of the group components affects the separation of the shared and group components, but
does not affect the individual error $R$. This is expected since the separation of individual components occurs in the first stage and only involves the sum of the shared and group components, which is not affected by the correlation of the additive terms defining it.

On the other hand, in the $s_{u, u}$ plot, we can observe that the similarity of the individual components affects the errors of all other components.
This is also expected as the separation of the individual components occurs in the first stage and thus propagates the errors into the second stage, where separation of shared and group components takes place. A smaller change in the group error $Q$ compared to the shared $S$ may be explained by the fact that the estimation of each group component depends only on the individual components within its group, while the estimation of the shared component depends on all individual components. Notice that this can also be seen from our theoretical bounds in Theorem \ref{main_consistency_theorem}, which states that the error of $S$ depends on the maximal similarity $s_{u, u}$ of all individual components and the error of $Q_k$ depends on $s_{u, u}^{(k)}$, the maximal similarity of individual components in group $k$.


\subsection{Results for the binary  edge model}\label{logistic_compar_appendix_section}

Here, we repeat the methods comparison experiment in Section \ref{compar_other_methods_exp_section} under the logistic edge model. Figure \ref{fig:errors_across_models_logistic} presents the corresponding results. We omit results for $n=100$ and $M=8$ due to the instability of the Logistic model fitting for small samples.

\begin{figure}[t!]
    \centering
    \includegraphics[width=0.95\linewidth]{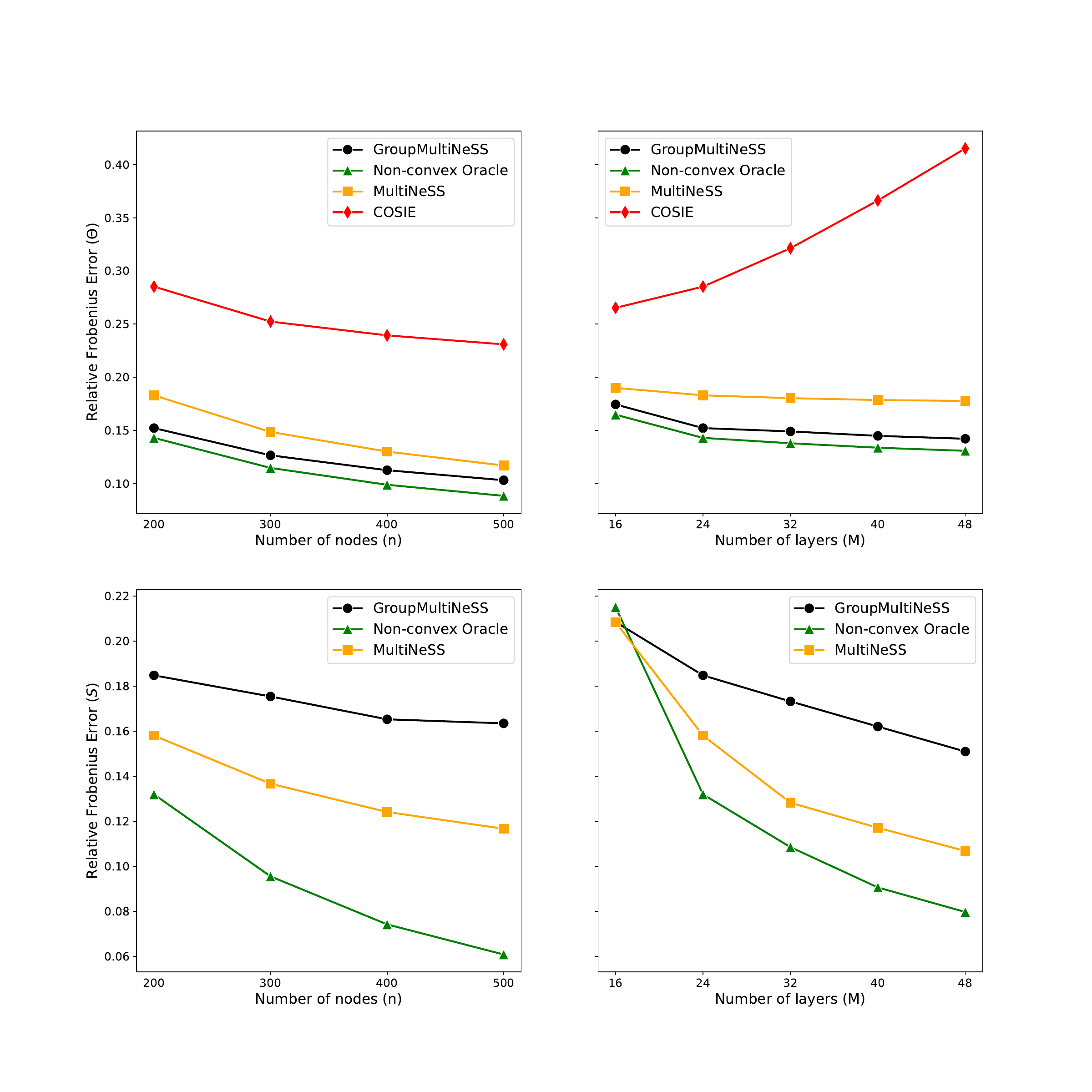}
    \caption{Change in ARFE of $\Theta$ (top row) and $S$ (bottom row) with the increase of (left column) the number of nodes $n$ with $M=16$ and (right column) the number of layers $M$ with $n=200$. In all simulations, the number of groups is $K=4$ and the latent dimension is $d=3$. Layers are sampled from the Logistic edge-entry model.  }
    \label{fig:errors_across_models_logistic}
\end{figure}

When it comes to estimating the overall parameter $\Theta$ (top row of Figure \ref{fig:errors_across_models_logistic}), GroupMultiNeSS is again very close to the oracle  and outperforms both MultiNeSS and COSIE in all regimes;  MultiNeSS also does substantially better than COSIE.  However, in the estimation of the shared component $S$ (bottom row) there is a bigger gap between GroupMultiNeSS and the oracle, and in fact MutliNeSS outperform GroupMultiNeSS.   This is likely due to the fact that the non-linear link function leads to an adjacency matrix with a much higher estimated rank than that of the original latent space, leading to a larger number of noise eigenvalues inflated during the refitting step.  This suggests more careful tuning is needed for the logistic link function model, and possibly a modification of the refitting step;  we leave this topic for future work.  





\end{document}